\documentclass[journal,11pt,twocolumn,romanappendices]{IEEEtran}
\ifCLASSINFOpdf
   \usepackage[pdftex]{graphicx}
\else
\fi
\usepackage{amssymb,bbm}

\usepackage{amsmath}
\usepackage{color}

\def \ver{techrep}
\def \verreview{review}
\def \vertechrep{techrep}

\newcommand{\changes}[1]{%
	\ifx\ver\verreview%
		{\textcolor{blue}{#1}}%
	\else {#1}%
	\fi%
}

\newcommand{\citeapp}[1]{%
	\ifx\ver\vertechrep%
		{Appendix \ref{#1}}%
	\else {\cite[Appendix \ref{#1}]{implicit_arxiv}}%
	\fi%
}

\ifx\ver\vertechrep%
	\def\myFigureScale{0.8}%
\else \def\myFigureScale{0.6}%
\fi%

\DeclareMathOperator*{\argmax}{arg\,max}

\newcommand{\Pa}{\mathcal P}
\newcommand{\x}{\mathbf x}
\newcommand{\s}{\mathbf s}
\newcommand{\X}{\mathbf X}
\newcommand{\y}{\mathbf y}

\newcommand{\E}{{\rm I\kern-.3em E}}

\newcommand{\lru}{\textproc{LRU}}
\newcommand{\lruone}{\textproc{LRU-One}}
\newcommand{\qlruone}{\textproc{qLRU-One}}
\newcommand{\fifoone}{\textproc{FIFO-One}}
\newcommand{\fifoblind}{\textproc{FIFO-Blind}}
\newcommand{\fifolazy}{\textproc{FIFO-Lazy}}
\newcommand{\qlrublind}{\textproc{qLRU-Blind}}
\newcommand{\lruall}{\textproc{LRU-All}}

\newcommand{\qlru}{\textproc{qLRU}}
\newcommand{\qlrulazy}{\textproc{qLRU-Lazy}}
\newcommand{\fifo}{\textproc{FIFO}}
\newcommand{\random}{\textproc{RANDOM}}
\newcommand{\twolru}{\textproc{2LRU}}
\newcommand{\Klru}{\textproc{kLRU}}
\newcommand{\twolrulazy}{\textproc{2LRU-Lazy}}
\newcommand{\twolrublind}{\textproc{2LRU-Blind}}

\newcommand{\blind}{\emph{blind}}
\newcommand{\lazy}{\emph{lazy}}
\newcommand{\one}{\emph{one}}
\newcommand{\all}{\emph{all}}

\usepackage{amsthm}
\theoremstyle{plain}
\newtheorem{thm}{Theorem}[section]
\newtheorem{lem}[thm]{Lemma}
\newtheorem{prop}[thm]{Proposition}

\theoremstyle{definition}

\theoremstyle{remark}

\usepackage{algpseudocode,algorithm,algorithmicx}

\usepackage{url}

\ifCLASSOPTIONcompsoc
  \usepackage[caption=false,font=normalsize,labelfont=sf,textfont=sf]{subfig}
\else
  \usepackage[caption=false,font=footnotesize]{subfig}
\fi
\usepackage{epstopdf}
\usepackage{verbatim}
\usepackage{algpseudocode}

\begin{document}
%

\title{Implicit Coordination of Caches in Small Cell Networks under Unknown Popularity Profiles} 


%
\author{
\IEEEauthorblockN{Emilio Leonardi\IEEEauthorrefmark{2} and Giovanni Neglia\IEEEauthorrefmark{1}\\
}
\IEEEauthorblockA{\IEEEauthorrefmark{2}Politecnico di Torino, Italy, leonardi@polito.it}\\
\IEEEauthorblockA{\IEEEauthorrefmark{1}Universit\'e C\^ote d'Azur -  Inria, France, giovanni.neglia@inria.fr}
}



\maketitle

\begin{abstract}
We focus on a dense cellular network, in which a limited-size cache is available at every Base Station (BS).
In order to optimize the overall performance of the system in such scenario, where  a significant  fraction of the users is covered by several BSs,
a tight coordination among nearby caches is needed. 
To this end,   this paper   introduces a class of simple and fully distributed caching policies, which
require neither direct communication among  BSs, nor
 a priori knowledge of content popularity. 
 Furthermore, we propose a novel approximate analytical methodology  
to assess the performance of interacting caches under such policies.
Our approach builds upon the well known characteristic time approximation~\cite{che02} 
and provides predictions  that are surprisingly accurate (hardly distinguishable from the simulations)
in most of the scenarios.
Both synthetic and trace-driven results show that the our caching policies achieve excellent performance (in some cases provably optimal). 
They outperform state-of-the-art dynamic policies for interacting caches, 
and, in some cases, also the greedy content placement, which is
known  to be the best performing  polynomial  algorithm  under static and perfectly-known content popularity profiles.

\end{abstract}


\section{Introduction}
\label{s:intro}

In the last years, with the advent and the proliferation of mobile devices (smart-phones, tablets),
along with  a constant increase of the overall traffic flowing over Internet, we have assisted to a radical shift of the traffic at the edge,
from the  wired/fixed  segment of the network to the wireless/mobile segment.   This trend is expected to continue and intensify in the next few years.
According to CISCO forecasts \cite{CISCO} in the 5 years ranging from 2016 to 2021 traffic demand on the cellular network 
will approximately  increase by a factor 8. Such traffic increase may pose a tremendous stress on the wireless infrastructure and can be satisfied only 
by densifying the cellular network and redesigning its infrastructure.  To this end,  the integration of  caches at the edge of the cellular 
network can be effective to  reduce the load on   the back-haul links.  Caches, indeed, by moving contents closer to the user, 
can effectively contribute to ``localize'' the traffic, and to achieve:
i) the reduction of the load on the \changes{core network and back-haul links}; ii)~the reduction of   the   
latency perceived by the user.

In this paper we focus our attention on a dense cellular network, where caches are placed at every Base Station (BS) and a significant fraction of the users can be served (is ``covered'') by two or more BSs (whose cells are said to ``overlap'').  In this context, an important open question is  how to effectively  coordinate  different 
edge-caches,  so  to  optimize  the  global performance (typically the hit ratio, i.e. the fraction of
users' requests that are satisfied by local caches).
Given the possibility of partial cell overlap, the cache coordination scheme 
should, indeed,    reach    an optimal trade-off
between  two somewhat  conflicting  targets:  i)~make  top-popular  contents  
available  everywhere, so to maximize the population of users who can  retrieve them from local caches,
ii)  diversify the contents stored at overlapping cells,
so to maximize the number of contents available to users in the overlap.
Optimal solutions can be easily figured out  for  the two extreme cases: 
when cells do not overlap, every cache should be filled with the most popular contents;
when cells overlap completely,  every  cache  should be filled with a different set of contents.
In the most general case, however, finding the optimal content allocation strategy  requires the solution of  an NP-hard problem~\cite{shanmugam13}. 

In this paper, we propose a class of  fully distributed schemes to coordinate caches in cellular systems with partially overlapping cells,
so to maximize the overall hit ratio. Our policies are very simple, and, differently from most of the previous work, do not require any a priori knowledge of content popularity. In addition, we propose a novel  analytical approach  to accurately evaluate the performance of  such caching systems with limited computational complexity.

\subsection{Related work}
\label{s:related}
Due to the space constraints,  we limit 
our description to the work that  specifically addresses the  
caching problem  in dense cellular networks. 

To the best of our knowledge, the idea to coordinate content placement at caches located at close-by BSs was first proposed in~\cite{Caire12} and its extension \cite{shanmugam13} under the name of femto-caching. This work assumes that requests follow the Independent Reference Model (IRM) and  geographical popularity profiles are available, i.e. content requests are independent and request rates are known for all the cell areas and their intersections. The optimal content placement to maximize the hit ratio has been formulated in terms of  
an NP-hard combinatorial problem. A greedy heuristic algorithm was then proposed and its performance analyzed. In particular the algorithm  is shown to guarantee a $\frac{1}{2}$-approximation of the maximum hit ratio. 
In \cite{Poularakis14},  the authors  have generalized the approach of \cite{Caire12,shanmugam13}, providing a formulation for the joint content-placement and user-association problem that   maximizes the hit ratio. Efficient heuristic solutions have also been proposed. 
Authors of \cite{Naveen15}  have included the bandwidth costs  in the formulation,
and have proposed an on-line algorithm for the solution of the resulting problem. \cite{tuholukova17} considers the case when small cells can coordinate not just in terms of what to cache but also to perform Joint Transmission. 
In \cite{Chattopadhyay16}, instead, the authors have designed a  distributed 
algorithm based on Gibbs sampling, which is shown to asymptotically converge to the optimal allocation.
\cite{Anastasios2}~revisits the  optimal  content placement problem within a 
stochastic geometry  framework. Under the assumption that  both  base stations and users are points of two  homogeneous spatial Poisson processes,
it derives an elegant analytical characterization of the optimal policy and its performance. More recently, in \cite{avrachenkov17} the authors have developed a few   asynchronous 
distributed cooperative content placement algorithms with polynomial complexity
and limited  communication overhead (communication takes place only between overlapping cells), whose performance has been shown to be very good 
in most of the tested scenarios. 

We would like to emphasize that all the previously proposed schemes, differently from ours, rely on  the original assumption in~\cite{Caire12} that geographical content popularity profiles are  known by the system. Therefore we will refer to these policies as \lq \lq informed'' ones.

Reliable popularity estimates over small geographical areas may be very hard to obtain~\cite{leconte16},  because i) most of the contents are highly ephemeral, ii)  few users are located in a small cell, and iii) users keep moving from one cell to another. On the contrary, policies like 
\lru{} and its variants (\qlru, \textsc{2LRU}, \dots) do not rely on popularity estimation---we call them ``uninformed''---and are known to 
well behave under  time-varying popularities. For this reason they are a de-facto standard in most of the deployed caching systems.
\cite{giovanidis16} proposes a generalization of \lru{} to a dense cellular scenario. As above, a user at the intersection of multiple cells, can check the availability 
of the content at every covering cell and then download from one of it. 
The difference with respect  to standard \lru{} is how cache states are updated. In particular, the authors of~\cite{giovanidis16} consider two schemes: \lruone{} and \lruall. In \lruone{} each user is assigned to a reference cell/cache and only the state of her reference cache is updated upon a hit or a miss, independently from which cache the content has been retrieved from.\footnote{
	The verbal description of the policy in~\cite{giovanidis16} is a bit ambiguous, but the model equation shows that the state of a cache is updated by and only by the requests originated in the corresponding Voronoi cell, i.e.~from the users closest to the cache.
} In \lruall{} the state of all the caches covering the user is updated. These policies do not require communication among caches. Moreover, their analysis is relatively easy because each cache can be studied as an isolated one.
 Unfortunately, these policies typically perform significantly worse than informed schemes (see for example the experiments in \cite{avrachenkov17}).

\subsection{Paper Contribution}
\label{s:contri}
\changes{ This paper has four main contributions.}

\changes{First, we propose in Sec.~\ref{s:model} a novel approximate analytical approach  to study systems of interacting caches, under different caching policies.  
Our framework builds upon the well known  characteristic time  approximation~\cite{che02} for individual caches,
and, in most of the scenarios, provides predictions  that are surprisingly accurate (practically indistinguishable from the simulations, as shown in Sec~\ref{s:validation}).}

\changes{Second, we propose a class of simple and fully distributed  ``uniformed''  schemes that effectively coordinate  different caches in a dense cellular scenario in order  to maximize the overall  hit ratio of the caching system.   Our schemes represent an enhancement of  those proposed in~\cite{giovanidis16}. As \lruone{} and \lruall, our policies neither rely on popularity estimation, nor require communication exchange among the BSs. The schemes achieve \emph{implicit} coordination among the caches through specific cache update rules, which are driven by the users' requests. Differently from \cite{giovanidis16}, our update rules tend to couple the states of different caches. Despite the additional complexity, we show that accurate analytical evaluation of the system is still possible through our approximate analytical approach.}

\changes{Third, we rely on our analytical model to show that, under IRM, our policies can significantly outperform other uninformed policies like those in~\cite{giovanidis16}. Moreover, the hit ratios are very close to those offered by the greedy scheme proposed in \cite{Caire12} under perfect knowledge of popularity profiles.  
More precisely,  we can prove  that, under some geometrical assumptions, a variant of \qlru{} asymptotically converges to the optimal static content allocation, while  in more  general scenarios,  the same version of \qlru{}  asymptotically achieves a locally optimal configuration.}

\changes{Finally, in Sec.~\ref{s:realistic}, we carry on simulations using a  request trace from a major CDN provider and BS locations in Berlin, Germany. The simulations confirm qualitatively the model's results. Moreover, under the real request trace, our dynamic policies can sometimes outperform the greedy static allocation that knows in advance the future request rate. This happens even more often in the realistic case when future popularity needs to be estimated from past statistics. Overall, our results suggest that it is better to rely on uninformed caching schemes with smart update rules than on informed ones fed by estimated popularity profiles, partially contradicting some of the conclusions of~\cite{elayoubi15}.}

\section{Network Operation}
\label{s:operation}


We consider a set of $B$ base stations (BSs) arbitrarily located in a given region $R \subseteq \mathbf R^2$, each equipped with a local cache.
Our system operates as follows. When user $u$ has a request for content $f$, it broadcasts an inquiry message to the set of BSs ($I_u$) it can communicate with. The subset ($J_{u,f}$) of those BSs that have the content $f$ stored locally declare its availability to user $u$. If any local copy is available ($J_{u,f}\neq \emptyset$), the user sends an explicit request to download it to one of the BSs in $J_{u,f}$. Otherwise, the user sends the request to one of the BSs in $I_u$, which will need to retrieve it from the content provider.\footnote{
	\changes{This two-step procedure introduces some additional delay, but this is inevitable in any femtocaching scheme where the BSs need to coordinate to serve the content.}
} Different user criteria  can be defined to select the BS to download from; for the sake of simplicity, in this paper, we assume that the user selects uniformly at random one of them. 
However, the analysis developed in the next section extends naturally under  general selection  criteria.
The selected BS serves the request. 
Furthermore, an opportunely defined  subset of BSs in $I_u$  updates its cache according to a local caching policy,  like \lru{}, \qlru{}\footnote{
In the case of \qlru, the cache will move the content requested to the front of the queue upon a hit and will store the content at the front of the cache with probability $q$ upon a miss. \lru{} is a \qlru{} policy with $q=1$. 
} or \twolru,\footnote{\twolru{} uses internally two \lru{} caches:  one for the metadata, 
and the other for the actual contents  (for this reason we say that \twolru{} is a {two-stage} cache). 
Upon a miss, a content is stored in the second cache only if its metadata are already present in the first cache. See \cite{garetto16} for a more detailed description.}  etc.
The most natural  update rule is that only the cache serving the content updates its state independently from the identity of the user generating the specific request. We call this update rule \blind. 
\changes{At the same time, it is possible to decouple content retrieval from cache state update as proposed in~\cite{giovanidis16}.  For example each user $u$  may be statically  associated to a given BS,  whose state is updated upon every request from $u$ independently from which BS has 
served the content. We refer to this update rule as \one, because of the name of the corresponding policy proposed in~\cite{giovanidis16} (\lruone).}  Similarly, we indicate as \all{} the update rule where all the BSs in $I_u$  update their state upon a request from user $u$ (as in \lruall). These update rules can be freely combined with existing usual single cache policies, like \qlru, \lru, \twolru, etc., leading then to schemes like \lruone, \qlrublind, \textsc{2LRU-All}, etc., with obvious interpretation of the names.

The analytical framework presented in the next section allows us to study a larger set of update rules, where the update can in general depend on the identity of the user as well as on the current set $J_{u,f}$ of caches from which $u$ can retrieve the content.  When coupled with local caching policies, these update rules do not require any explicit information exchange among the caches, but they can be implemented by simply piggybacking the required information ($J_{u,f}$) to user $u$'s request. In particular, in what follows, we will consider the \lazy{} update rule, according to which 
\begin{enumerate}
\item only the cache serving the content may update its state,
\item but it does only if no other cache could have served the content to the user (i.e.~only if $|J_{u,f}|\le 1$).
\end{enumerate}
This rule requires only an additional bit to be transmitted from the user to the cache. We are going to show in Sec.~\ref{s:lazyqlru} that such bit is sufficient to achieve a high level of coordination among  different caches and, therefore, a high hit ratio. Because no communication among BSs is required, we talk about \emph{implicit coordination}. For the moment, the reader should not be worried if he/she finds the rationale behind \lazy{} obscure and can regard \lazy{} as a specific update rule  among many others possible.
\changes{Table~\ref{notation-summary} summarises the main notation used in this paper.}

\section{Model}
\label{s:model}

\begin{table}
\changes{
\begin{center}
\caption{Summary of the main notation}
\label{notation-summary}
\begin{tabular}{|l|l|}
\hline
Symbol &  Explanation \\
\hline
$t$ & time\\
$ u$ & generic user\\
$f $ & generic content\\
$b$ & generic cell \\
$F$  & number of files\\
$B$  & number of base stations\\
$C$     & cache size\\ 
$I_u$ & set of BSs communicating with $u$\\
$J_{u,f}$ & set of BSs able to provide $f$ to $u$ \\
$S_b$ & surface of cell b\\
$\mu(A)$ &  expected number of users in region $A$ \\
$\Lambda_f(A)$ & content $f$ request rate from region $A$ \\
$\X_f$   & configuration of content $f$ in caches \\
$x_f^{(b)}$ &  component $b$ of $\X_f$: $x_f^{(b)}\in \{0,1\}$  \\
$\x_f^{(-b)}$ &   configuration of content $f$ in all caches but $b$  \\
$T_c^{(b)}$ & characteristic time at cache $b$ \\
$T_{S,f}^{(b)}$ & sojourn time of content $f$ in cache $b$ \\
$\nu_f^{(b)}$ &  transition rate $1\to 0$ for $x_f^b$  given $\x_f^{(-b)}$\\
$\alpha_f^{(b)}$ & transition rate $0\to 1$ for  $x_f^b$ given $\x_f^{(-b)}$\\
$d(s,A)$ & distance between point $s$ and region $A$\\
\hline 
\end{tabular}
\end{center}
}
\end{table}

We assume that mobile users are spread over the region $R$ according to a Poisson point process with density $\mu()$, so that $\mu(A)$ denotes the expected number of users in a given area $A\subseteq R$. 
Users generate independent requests for $F$ possible contents. In particular, a given user requests  content $f$  according to a Poisson process with rate $\lambda_f$. It follows that the aggregate request process  for content $f$ from all the users located in $A$ is also a Poisson process with rate $\Lambda_f(A)=\mu(A) \lambda_f$. \changes{ For the sake of presentation, in what follows we will consider that users' density is constant over the region, so that $\mu(A)$ is simply proportional to the surface of $A$, but our results can be easily generalized.} 
Our analysis can also be extended  to a more complex content request model that takes into account temporal locality~\cite{traverso13} as we discuss in Sec.~\ref{s:temporal_locality}.
Contents (i.e. caching units) are assumed to have the same size. This assumption
can be justified in light of the fact that often contents  correspond to  chunks in which  larger files are broken. In any case, it is possible to extend the model, and most of the analytical results below, to the case of heterogeneous size contents.\footnote{
Similar results for \qlrulazy{} in Sec.~\ref{s:qlrulazy_convergence} hold if we let the parameter $q$ be inversely proportional to the content size as done in~\cite{neglia17tompecs}. 
} For the sake of simplicity, we assume that each cache is able to store $C$ contents.
Finally, let $S_b$ denote the coverage area of BS $b$. 

In what follows, we  first present some key observations for an isolated cache, and then we extend our investigation to interacting caches when cells overlap. We will first consider the more natural \blind{} update rule, according to which any request served by a BS triggers a corresponding cache state update. 
We will then discuss 
how to extend the model to other update rules in Sec.~\ref{s:update_rules}.

\subsection{A single cell in isolation}
\label{s:single_cell}
We start considering a single BS, say it $b$, with cell size $S_b$. The request rate per content $f$ is then $\Lambda_f^{(b)}(S_b)=\mu(S_b) \lambda_f$. We omit in what follows the dependence on $S_b$.

Our analysis relies on the now standard cache characteristic time approximation (CTA) for a cache in isolation, 
which is known to be one of the most effective approximate approaches for analysis of caching systems.\footnote{Unfortunately, the computational cost to exactly analyse even a single
LRU (Least Recently Used) cache, grows exponentially with both the  cache  size  and  the  number of  contents \cite{Dan90}.}
CTA was first introduced (and analytically justified) in~\cite{fagin77} and later rediscovered in~\cite{che02}. It was originally proposed for \lru{} under the IRM request process, and it has been later extended to different caching policies and different requests processes~\cite{garetto16,garetto15}. The characteristic time $T_c$ is the time a given content spends in the cache since its insertion until its eviction in absence of any request for it. In general, this time depends  in a complex way from the dynamics of other contents requests. Instead, the CTA assumes that $T_c$ 
 is  a random variable independent  from other contents dynamics and
with an assigned distribution (the same for every content). This assumption makes it possible to decouple the dynamics of the different contents: upon a miss for content $f$, the content is retrieved and a timer with random value $T_c$ is generated. When the timer expires, the content is evicted from the cache. Cache policies differ for i) the distribution of $T_c$ and ii) what happens to the timer upon a hit. For example, $T_c$ is a constant under \lru, \qlru, \twolru{} and \fifo{} and exponentially distributed under \random. Upon a hit, the timer is renewed under 
\lru, \qlru{} and \twolru, but not under \fifo{} or \random. Despite its simplicity, CTA  was shown  to provide asymptotically exact predictions for a single LRU cache under IRM  as the cache size grows large~\cite{fagin77,Jele99,fricker2012}.

What is important for our purposes is that, once inserted in the cache, a given content $f$ will sojourn in the cache for a random amount of time $T_{S,f}$, that can be characterized for the different policies. In particular, if the timer is not renewed upon a hit (as for \fifo{} and \random{}), it holds:
\[T_{S,f}^{(b)}=T_c^{(b)},\]
while if the timer is renewed, it holds:
\[T_{S,f}^{(b)}=\sum_{k=1}^M Y_k+ T_c^{(b)},\] 
where $M \in \{0,1, \dots\}$ is the number of consecutive hits preceding a miss and \changes{$Y_k$ is the time interval between the $k$-th hit and the previous content request}.  For example, in the case of \lru{} and \qlru{}, $M$ is distributed as a geometric random variable with parameter $p=1-e^{-\Lambda_f^{(b)} T_c^{(b)}}$, and $\{Y_k\}$ are i.i.d. truncated exponential random variables over the interval $[0,T_c^{(b)}]$.

We denote by $1/\nu_f^{(b)}$ the expected value of $T_{S,f}^{(b)}$, that is a function of the request arrival rate $\Lambda_f^{(b)}$.
\begin{equation}
\label{e:rate}
\frac{1}{\nu_f^{(b)}\!\!\left(\Lambda_f^{(b)}\right)}\triangleq \E[T_{S,f}^{(b)}].
\end{equation}
For example it holds:
\begin{itemize}
	\item[]\fifo, \random:\hspace{2mm}  $\nu_f^{(b)}\!\!\left(\Lambda_f^{(b)}\right)=1/ T_c^{(b)}$
	\item[]\lru, \qlru:\hspace{10mm} $\nu_f^{(b)}\!\!\left(\Lambda_f^{(b)}\right)=\frac{\Lambda_f^{(b)}}{e^{\Lambda_f^{(b)} T_c^{(b)}}-1}$.
\end{itemize}
\changes{where the last  expression  can be obtained through standard renewal arguments (see~\cite{choungmo14}).}


Let $X^{(b)}_f(t)$ be the process indicating whether content $f$ is in the cache $b$ at time $t$. For single-stage caching policies, such as  \fifo, \random, \lru{} and \qlru,  $X^{(b)}_f(t)$ is an ON/OFF renewal process with ON period distributed as \changes{$T_{S,f}^{(b)}$} and OFF period distributed exponentially with mean value $1/\alpha_f^{(b)}$, where $\alpha_f^{(b)}=\Lambda_f^{(b)}$ for \fifo{} and \random{},  and $\alpha_f^{(b)}=q \Lambda_f^{(b)}$ for  \qlru. The process $X^{(b)}_f(t)$ can also be considered as the busy server indicator of an $M/G/1/0$ queue with service time distributed as $T_{S,f}^{(b)}$.\footnote{ Under CTA a cache with capacity $C$ becomes indeed equivalent to a set of $C$ parallel independent $M/G/1/0$ queues, one for each content.
}  This observation is important because the stationary distribution of $M/G/n/0$ queues  depends on the service time only through its mean~\cite{wolff1989stochastic}. As a consequence, for any metric depending only on the stationary distribution, an $M/G/1/0$ queue is equivalent to an $M/M/1/0$ queue with service time exponentially distributed with the same service rate $\nu_f^{(b)}$. In particular,  the stationary occupancy probability $h_f^{(b)}=\Pr\{X^{(b)}_f(t)=1\}$ is simply $h_f^{(b)}=\frac{\alpha_f^{(b)}/\nu_f^{(b)}}{\alpha_f^{(b)}/\nu_f^{(b)}+1}$.
Under CTA, the characteristic time $T_c$ can then be obtained by imposing that 
\begin{equation}
\label{e:constraint}
\sum_{f=1}^F h_f^{(b)} =C,
\end{equation}
for example using the bisection method.

The possibility of representing a cache in isolation with  an $M/M/1/0$ queue, i.e., as a simple continuous time Markov chain, does not provide particular advantages in this simple scenario, but it allows us to accurately study the more complex case when cells overlap,
 and users' request may be served by multiple caches.

\begin{figure}[htbp]
   \centering
   \hspace{-2em}
   \includegraphics[width=0.4\linewidth]{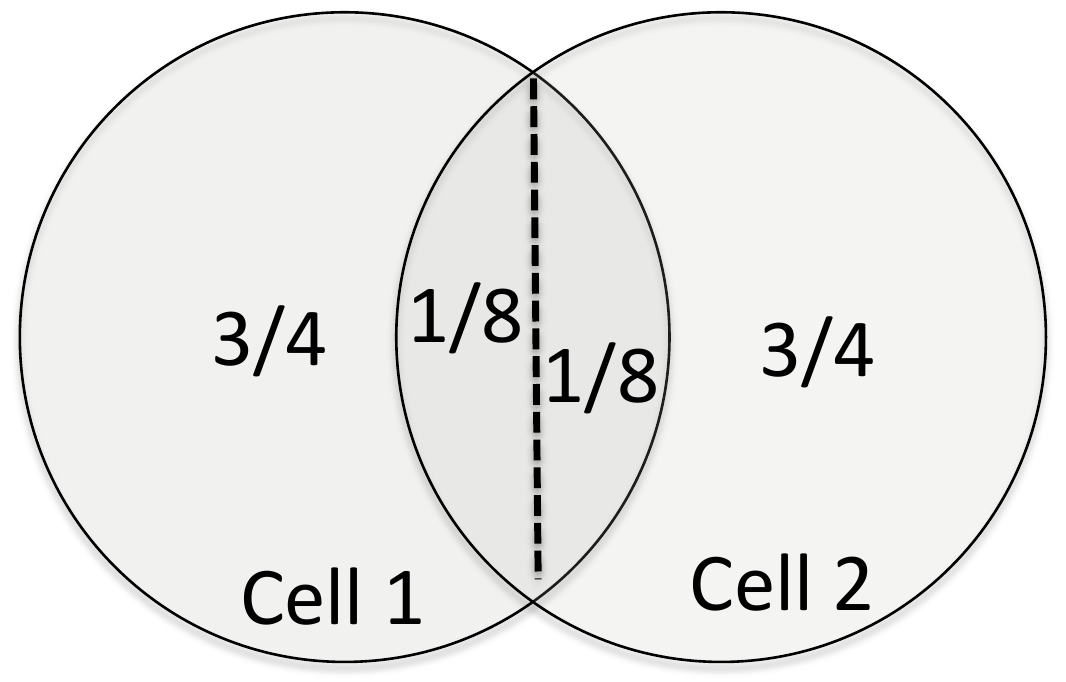}
   \caption{Two overlapping cells each with unit surface. The area of the overlapping area is $1/4$.} 
   \label{f:two_cells}
\end{figure}

\subsection{Overlapping cells}\label{overlapping-cells}
We consider now the case when $B$ cells may overlap.  
Let $X^{(b)}_f(t)$ indicate whether the BS $b$ stores at time $t$ a copy of content $f$ and $\X_f(t)=\left(X^{(1)}_f(t), \dots X^{(B)}_f(t)\right)$ be the 
vector showing where the content is placed within the  network. 
In this case the request rate seen by any BS, say it BS $b$, depends on the availability of the content at the neighbouring BSs, i.e. $\Lambda^{(b)}_f=\Lambda^{(b)}_f(\X_f(t))$. For example, with reference to the Fig.~\ref{f:two_cells}, if $\lambda_f=1$, BS~$1$ experiences a request rate for content $f$ equal to i) $1$ if it is the only BS to store the content, ii) $7/8$ if both BSs store the content or none of them does, iii)  $3/4$ if only  BS $2$ stores the content.

Our analysis of this system is based on the following approximation:\footnote{
	For any vector $\x$, we denote by $\x^{(-b)}$ the subvector of $\x$ including all the components but the $b$-th one and we can write $\x$ as $(x^{(b)},\x^{(-b)})$.
}
\begin{itemize}
	\item[A$1$] The stochastic process $\X_f(t)$ is a continuous-time Markov chain.  For each $f$ and $b$ the transition rate $\nu_f^{(b)}$ from state $\X_f(t)=(x_f^{(b)}=1, \x_f^{(-b)})$ to $(x_f^{(b)}=0, \x_f^{(-b)})$ is given by \eqref{e:rate} with $\Lambda_f^{(b)}$ replaced by $\Lambda^{(b)}_f(\X_f(t))$.
\end{itemize}
Before discussing the quality of approximation A$1$, let us first describe how it allows us to study the cache system. For a given initial guess of the characteristic times at all the $B$ caches, we determine the stationary distribution of the Markov Chains (MCs) $\X_f(t)$. 
We then compute the expected buffer occupancy at each cache and check if the set of constraints~\eqref{e:constraint} is satisfied. We then iteratively modify the vector of characteristic times by reducing (/increasing) the value for those caches where the expected buffer occupancy is above (/below) $C$.  Once the iterative procedure on vector  of characteristic times has reached convergence,
we compute the hit ratios for each content at each cache.

A$1$ envisages to replace the original stochastic process, whose analysis appears prohibitive, with a (simpler) MC. This has no impact 
on  any system metric that depends only on the stationary distribution in the following cases:
\begin{enumerate}
	\item isolated caches (as we have shown in Sec.~\ref{s:single_cell}),
	\item caches using \random{} policy, because the corresponding sojourn times coincide with the characteristic times and are exponentially distributed, hence A$1$ is not an approximation,	\item caches using \fifo{} policy under the additional condition in Proposition~\ref{p:insensitivity} below.
\end{enumerate}
In all these cases CTA is  the only approximation  having an impact on  the accuracy of model  results. In the most general case, however, A$1$ introduces an additional approximation. However our numerical evaluation shows that our approach provides very accurate results in all the scenarios we tested. 

We end this section by detailing the insensitivity result for a system of \fifo{} caches.
\begin{prop}
\label{p:insensitivity}
For \fifo{}, the probability of being in state $\x_f$ is insensitive to the distribution of the sojourn times $T_{S,f}^{(b)}$ as far as the Markov chain $\X_f(t)$ in approximation~A1 is reversible.
\end{prop}
 The proof of proposition~\ref{p:insensitivity} is in \citeapp{a:fifo_insensitivity} and relies on some insensitivity results for Generalized Semi Markov Processes. The reversibility hypothesis is for example satisfied for the cell trefoil topology considered in Sec.~\ref{s:validation} when users' density is constant.

\subsection{Model complexity}
Note that, in general,  the number of states of the Markov Chain describing the dynamics of  $\X_f(t)$   grows exponentially with  the number of cells $B$ (actually, it is equal to $2^{B}$), therefore modeling scenarios with a large number of cells becomes challenging and requires the adoption of efficient approximate techniques for the solution of probabilistic graphical methods~\cite{pelizzola}. 
However, scenarios with up to 10-12 cells can be efficiently modeled.  Furthermore, when the geometry exhibits some symmetry, some state aggregation becomes possible. For example, in the cell trefoil topology presented in the next section,  the evolution of $\X_f(t)$ can be represented by a  reversible birth-and-death Markov Chain with $(B+1)$ states ($\ll 2^{B}$). 

\subsection{Different Update rules}
\label{s:update_rules}
In presenting the model above, we have referred to the simple \blind{} update rule. 
Our modeling framework, however, can easily  accommodate  other update rules. For example for \one, if the reference BS is the closest one,   we should set \changes{$\Lambda^{(b)}_f=\lambda_f \mu\left(\{\s \in R; d(\s,S_b) \le d(\s,S_{b'}), \forall b'\}\right)$}, where $d(\s,A)$ denotes the distance between the point $\s$ and the set $A$. 
On the contrary, for \all, any request that could be served by the base station is taken into account, i.e. $\Lambda^{(b)}_f=\lambda_f \mu\left(S_b\right)$.
Finally, for  \lazy{} we have:
\begin{equation}
\label{e:lazy_rate}
\Lambda_f^{(b)}(\X_f)=\lambda_f \mu\!\!\left(S_b \setminus \bigcup_{b' | X_f^{(b')}=1} S_{b'}\right),
\end{equation}
i.e. only requests coming from areas that cannot be served from any other cache, affect the cache state. For example, with reference to Fig.~\ref{f:two_cells}, assuming $\lambda_f=1$, the request rate that contributes to update cache $1$ status is  $3/4$ when content $f$ is stored also at cache $2$. 

As we are going to discuss in Sec.~\ref{s:lazyqlru}, the update rules have a significant impact on the performance and in particular the lazy policies often outperform the others. Because our analysis will rely on the model described in this section, we first present in Sec.~\ref{s:validation} some validation results to convince the reader of its accuracy.

\subsection{Extension to multistage caching policies: \Klru}
The previous model can be extended to \twolru\ (and \Klru)  by following the approach proposed in~\cite{garetto16}.
In particular dynamics of the two stages can be represented by two  separate continuous time MCs whose states $\X_f^{(1)}(t)$ and $\X_f^{(2)}(t)$ correspond to the  configuration of content~$f$  at time $t$ in the system of virtual caches and  physical caches, respectively.
The dynamics of the system of virtual caches   at the first stage $\X_f^{(1)}(t)$ are not impacted by the presence of the second  stage and  perfectly emulate  
the dynamics of \lru\ caches;  therefore we  model them  by using the same MC as for  \lru .
On the contrary,  dynamics at the second stage depend on the first stage state. In particular content $f$ is inserted in the physical cache at the second stage upon a miss, only if the incoming request finds the content metadata within the first stage cache.
Therefore the transition rate from state $\X^{(2)}_f(t)=(x_f^{(b,2)}=0, \x_f^{(-b,2)})$ to $(x_f^{(b,2)}=1, \x_f^{(-b,2)})$ is given by 
{$\lambda_f^{(b,2)}= \Lambda^{(b)}_f(\X_f^{(2)}(t)) h_f^{(b,1)}$}, where $h_f^{(b,1)}$ represents the probability that  content $f$ metadata
is stored at the first stage.
Along the same lines the model can be easily extended to \Klru\ for $k>2$.

\subsection{How to account for temporal locality}
\label{s:temporal_locality}
Following the approach proposed in \cite{traverso13,garetto15}, we model the request process of every content~$f$ as a Markov Modulated Poisson Process (MMPP),  whose  modulating MC is a simple  ON-OFF MC.  
Now focusing, first, on a single cell scenario, we denote by  $\Lambda_f^{(b)}$ the aggregate arrival rate  of content $f$ 
at BS $b$ during ON periods. The arrival rate of content $f$ is, instead, null  during OFF periods. Let
$T^{\text{ON}}_f$ and $T^{\text{OFF}}_f$ denote the average sojourn times in state ON and OFF, respectively.\footnote{\changes{Sojourn times in both  states are exponentially distributed.}}
The idea behind this model is that each content has a finite lifetime  with mean $T^{\text{ON}}_f$ and after a random time with mean $T^{\text{OFF}}_f$, a new content with the same popularity arrives in the system. For convenience this new content is denoted by the same label $f$ (see~\cite{traverso13,garetto15} for a deeper discussion).
We can model the dynamics of content $f$ in the cache 
as an MMPP/M/1/0 queue with state-dependent service rate.
In particular service rates upon ON 
($\nu_f^{(b,ON)}$) are computed according to \eqref{e:rate}.  Service rates  on state OFF are simply set to 
$\nu_f^{(b,OFF)}=\frac{1}{T_c^{(b)}}$, as result of the application of   \eqref{e:rate} when the arrival rate of content-$f$ requests 
tends to 0.

The extension to the case of multiple overlapping cells can be carried out along the same lines of Section \ref{overlapping-cells}, 
(i.e. by applying approximation A$1$). As in \cite{garetto15},  the ON-OFF  processes governing  content-$f$
request rate at different cells are assumed to be 
perfectly synchronized (i.e., a unique underlying ON-OFF Markov Chain  determines content-$f$ request rate at every cell).
The resulting stochastic process $\X_f(t)$ is a continuous-time Markov Chain with $2^{B+1}$ states.

\section{Model validation}
\label{s:validation}

\begin{figure}[tbp]
	\centering
         \subfloat[]{\includegraphics[width=0.35\linewidth]{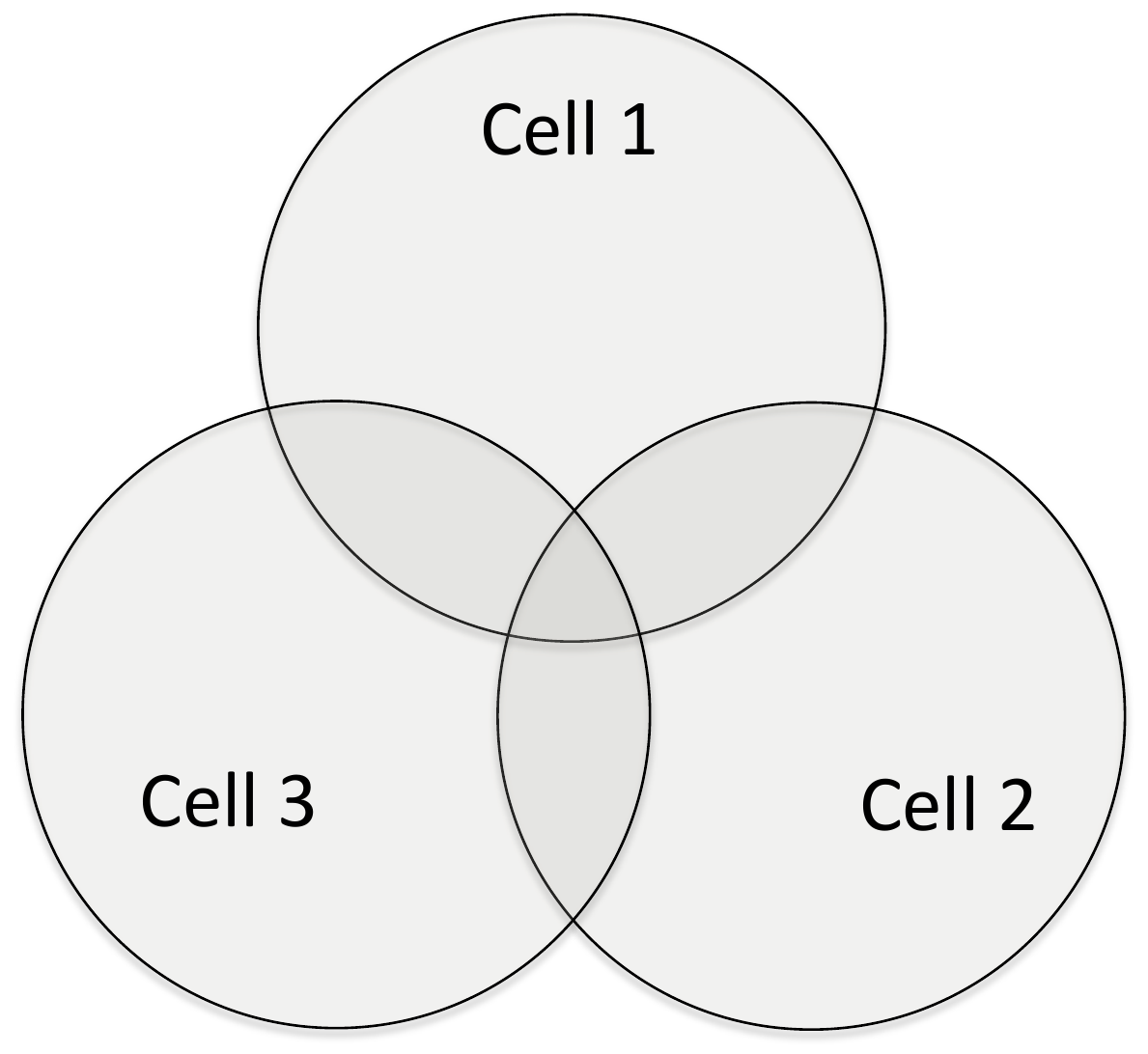}} \hspace{1cm}
         \subfloat[]{\includegraphics[width=0.35\linewidth]{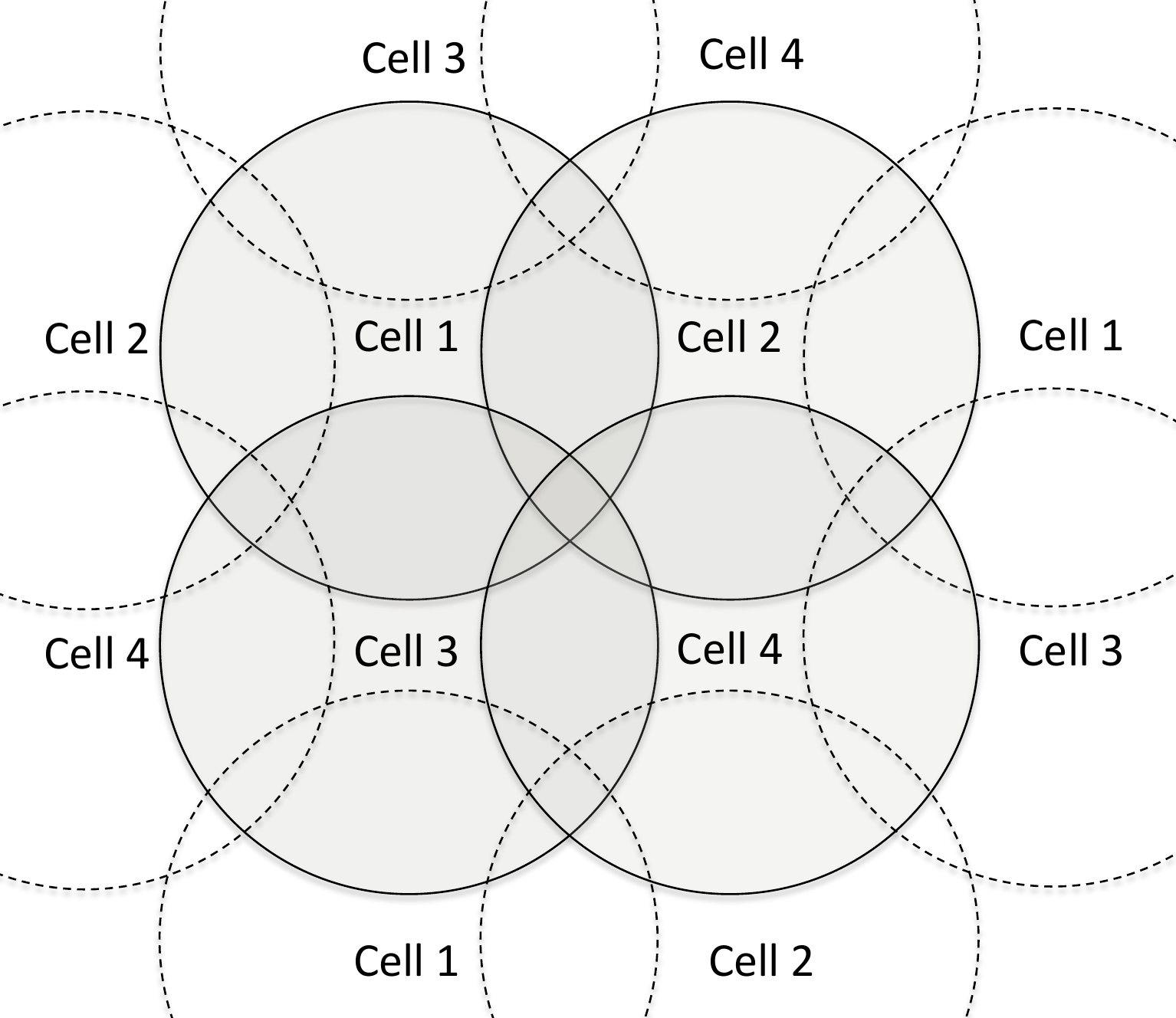}}\\    
         \caption{(a):  a trefoil. \hspace{18 mm} (b): a two-by-two cell torus.}
        \label{f:flower_torus}
        \vspace{-2 mm}
\end{figure}

\begin{figure*}[htbp]
         \centering
         \subfloat[One]{\includegraphics[width=0.33\linewidth]{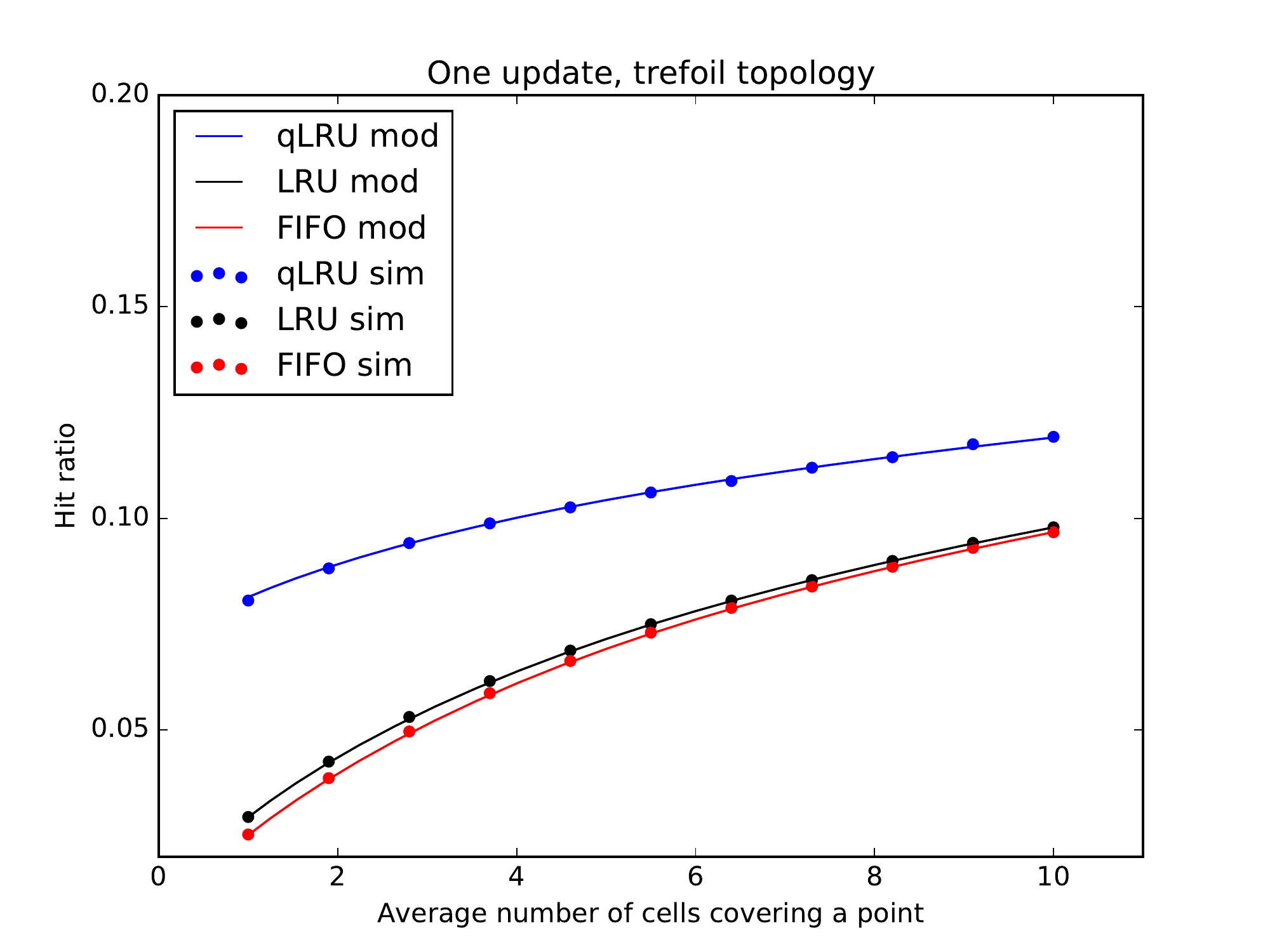}}\hfill
          \subfloat[Blind]{\includegraphics[width=0.33\linewidth]{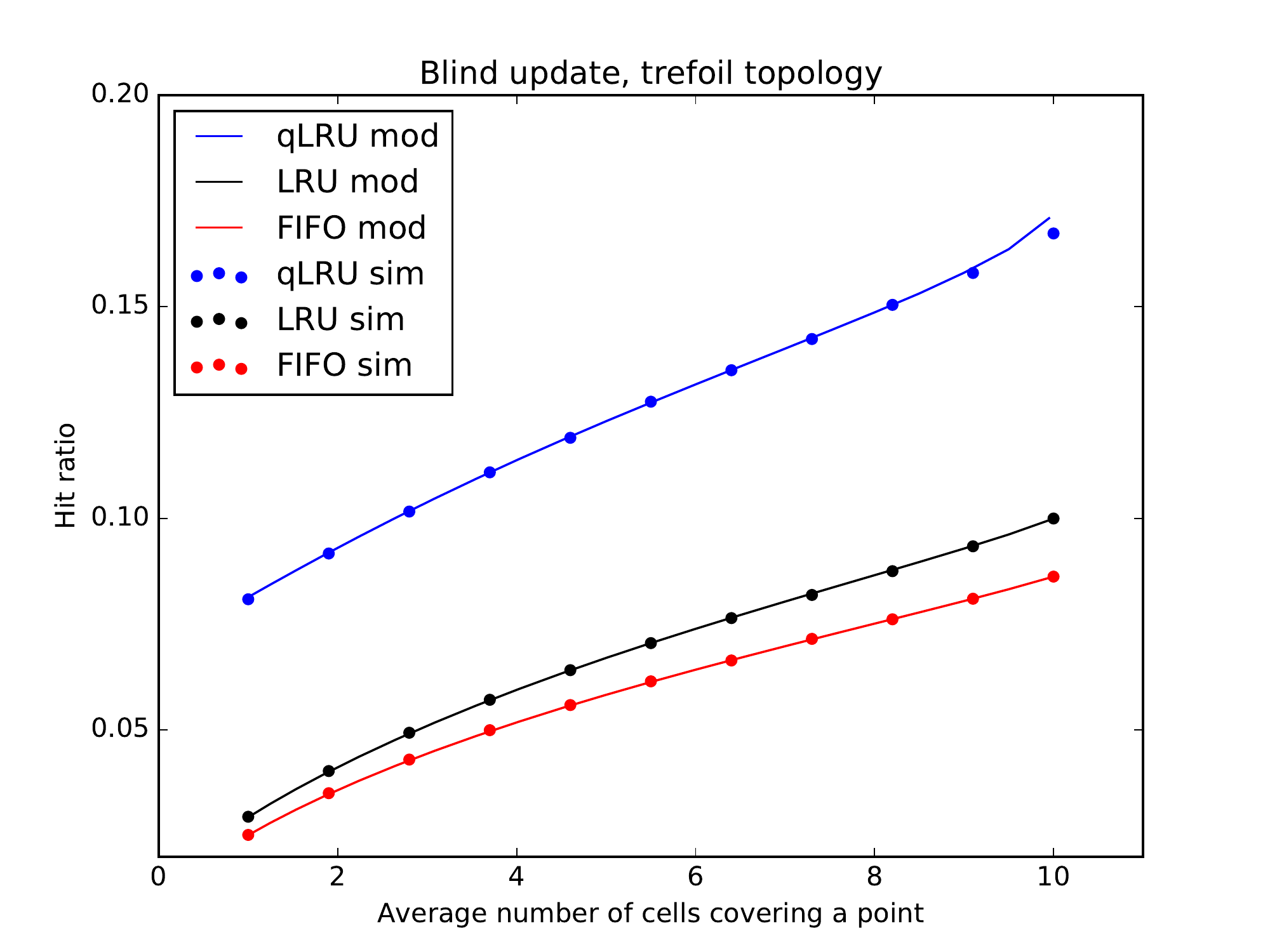}}\hfill
         \subfloat[Lazy]{\includegraphics[width=0.33\linewidth]{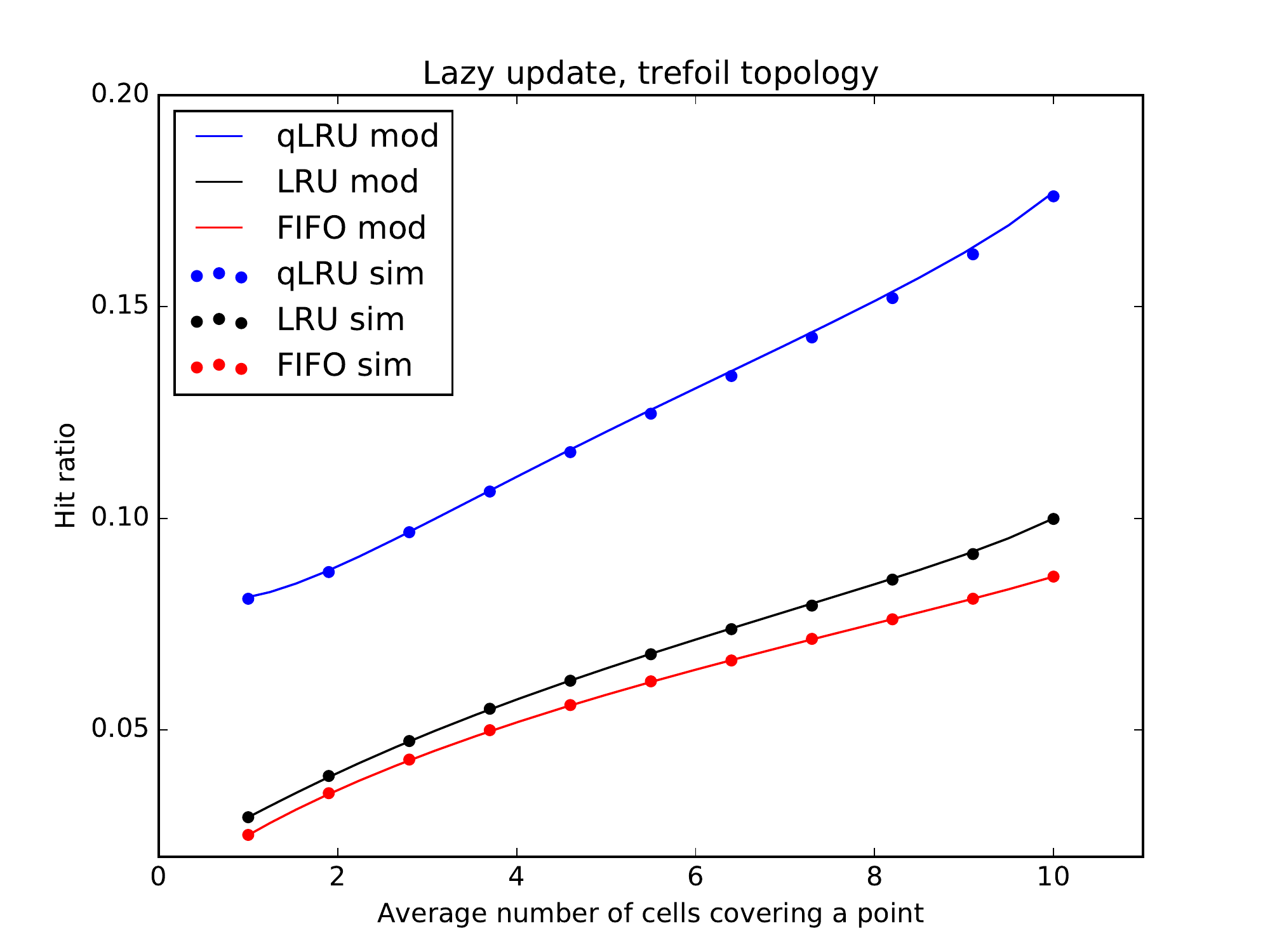}}\\
         \caption{\changes{Comparison between model predictions and simulations;  trefoil topology with $10$ cells; $C=100$; IRM traffic model with $\alpha=0.8$; \qlru{} employs $q=0.01$.}}
         \label{f:flower_val}
       \end{figure*}

In this section we validate our model by comparing its prediction against  simulation results for two different topologies.
Our trace-driven simulator developed in Python reproduces the exact dynamics of the caching system, and therefore can be used to test the impact of model assumptions (CTA and A1) on the accuracy of results in different traffic scenarios.
\changes{We start introducing a topology  which  exhibits a complete cell
symmetry (i.e. the hit rate of any allocation is invariant under cell-label permutations). 
In such a case, $\X_f(t)$  turns out to be a reversible Markov Chain.
Fig.~\ref{f:flower_torus}~(a) shows an example for $B=3$.  Generalizations for $B>3$ can be defined in a $B$ dimensional Euclidean-space by considering $B$ hyperspheres centered at the vertices of a regular simplex, but also on the plane if users' density is not homogeneous. We refer to this topology as the \emph{trefoil}.}
Then  we consider  a \emph{torus topology}  in which the base stations are disposed according to a regular grid on a torus as in Fig.~\ref{f:flower_torus}~(b).
For simplicity, in what follows, we assume that all the cells have the same size and a circular shape. 

Users are uniformly distributed over the plane and they request contents from a catalogue of $F=10^6$ contents whose popularity is distributed according to a Zipf's law with exponent $s=0.8$. Each BS can store up to $C=100$ contents. We have also performed experiments with $C=1000$ and $s=0.7$, but the conclusions are the same, so we omit them due to space constraints.


In Fig.~\ref{f:flower_val} we show the global hit ratio for different values of cell overlap in a trefoil topology with $10$ cells. The overlap is expressed in terms of  \changes{the expected number of caches a random user could download the content from}. The subfigure (a) shows the corresponding curves for \changes{ \fifoone{} and }\qlruone{} with $q=0.01$ and with $q=1$, which coincides with \lruone. The other subfigures are relative to the update rules \blind{} and \lazy.\footnote{
	\changes{We do not show results for \random{} or the update rule \all. \random{} is practically indistinguishable from \fifo{}  and \all{} was shown to have worse performance than \one{} for IRM traffic already in~\cite{giovanidis16}.}
} \changes{\fifoblind{} and \fifolazy{} coincide because in any case \fifo{} does not update the cache status upon a hit.}
The curves show an almost perfect matching between the results of the model described in Sec.~\ref{s:model} and those of simulation. 
Figure~\ref{f:torus_val} confirms the accuracy of the model also for the \changes{torus topology with 9 cells.}
Every model point has requested  less than 3 seconds of CPU-time on a   
INTEL Pentium G3420 @3.2Ghz for the cell trefoil topology, and  less than $5$ minutes for the torus.

\begin{figure*}[htbp]
         \centering
         \subfloat[One]{\includegraphics[width=0.33\linewidth]{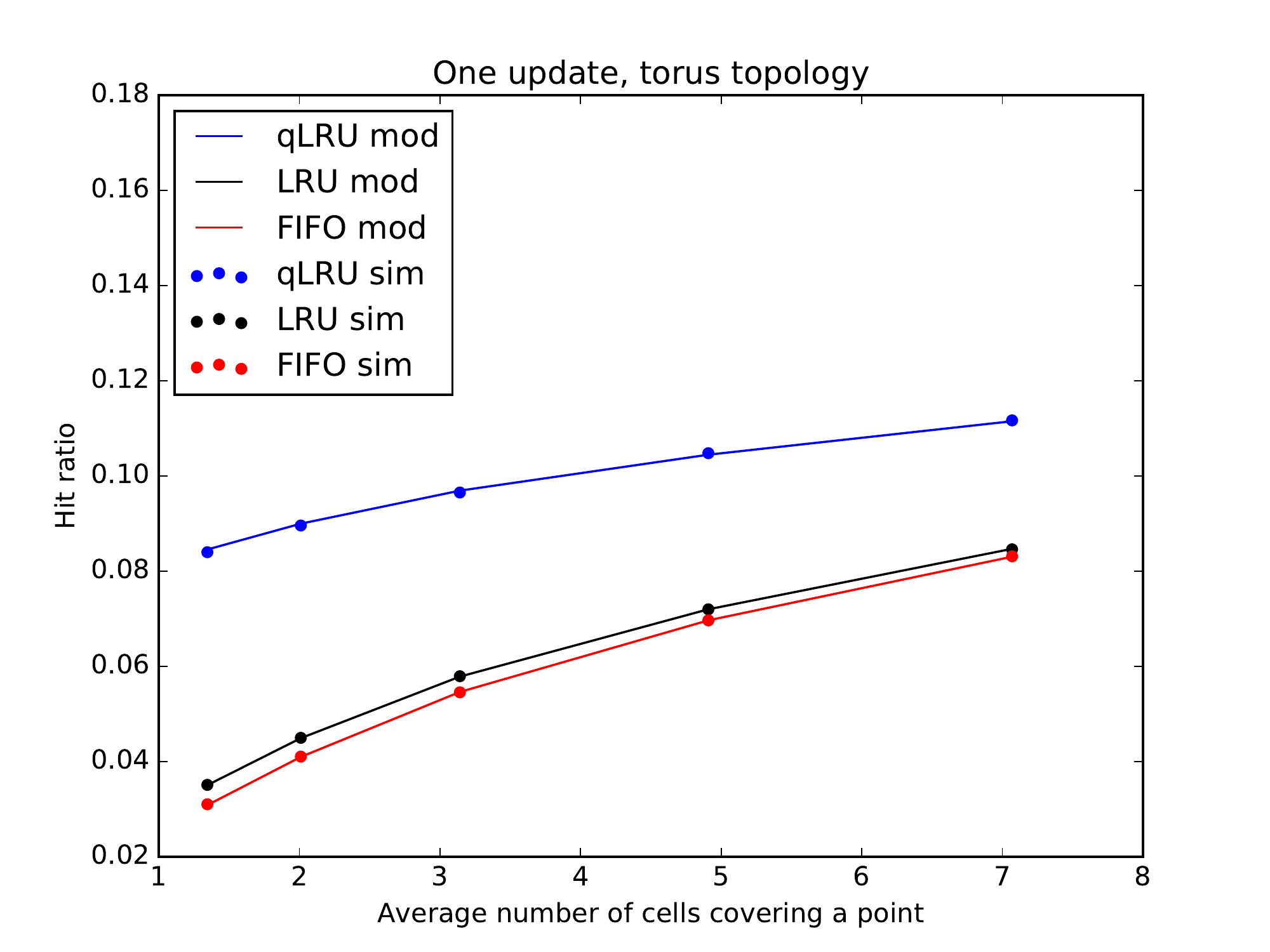}}\hfill
          \subfloat[Blind]{\includegraphics[width=0.33\linewidth]{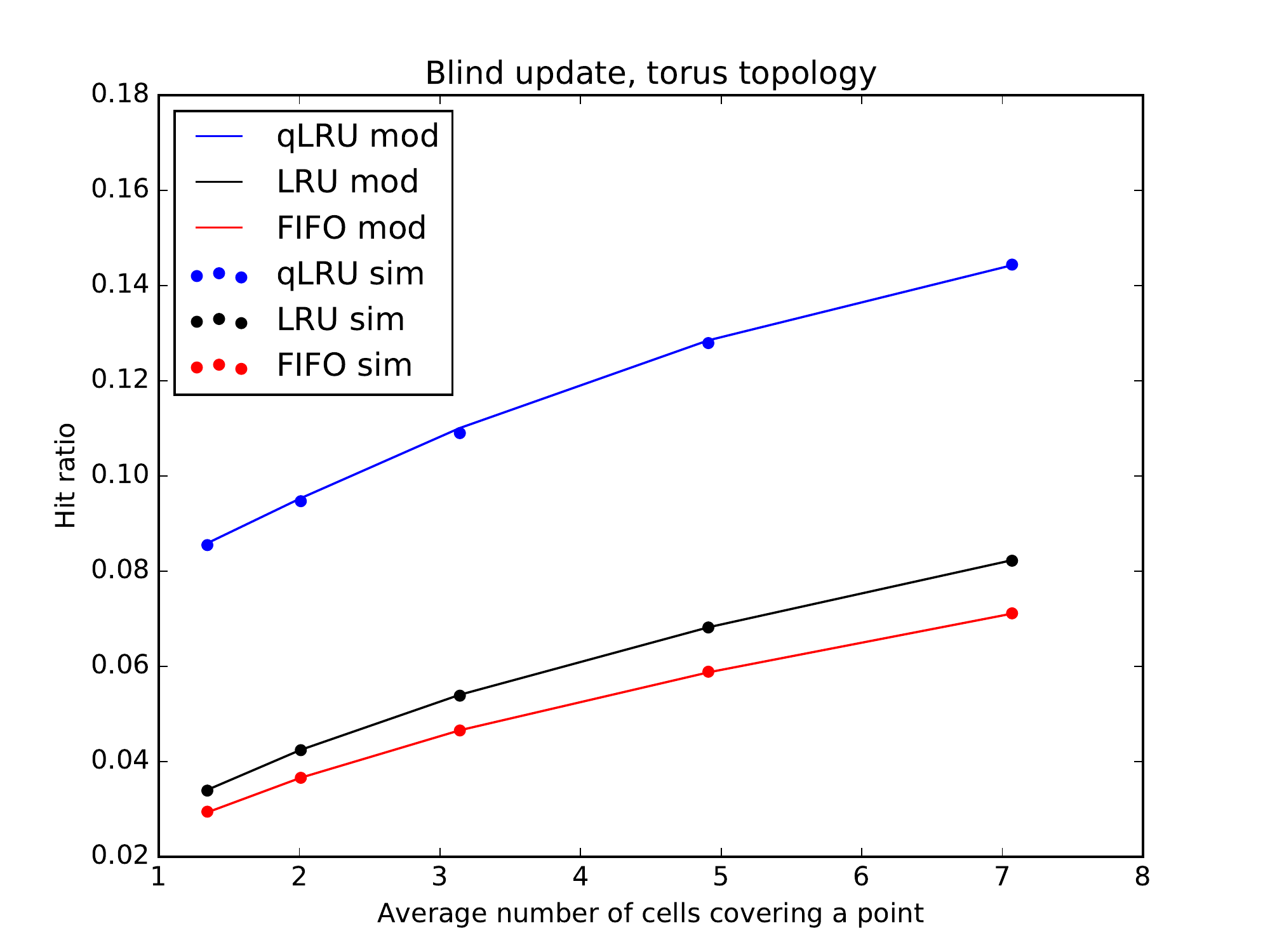}}\hfill
         \subfloat[Lazy]{\includegraphics[width=0.33\linewidth]{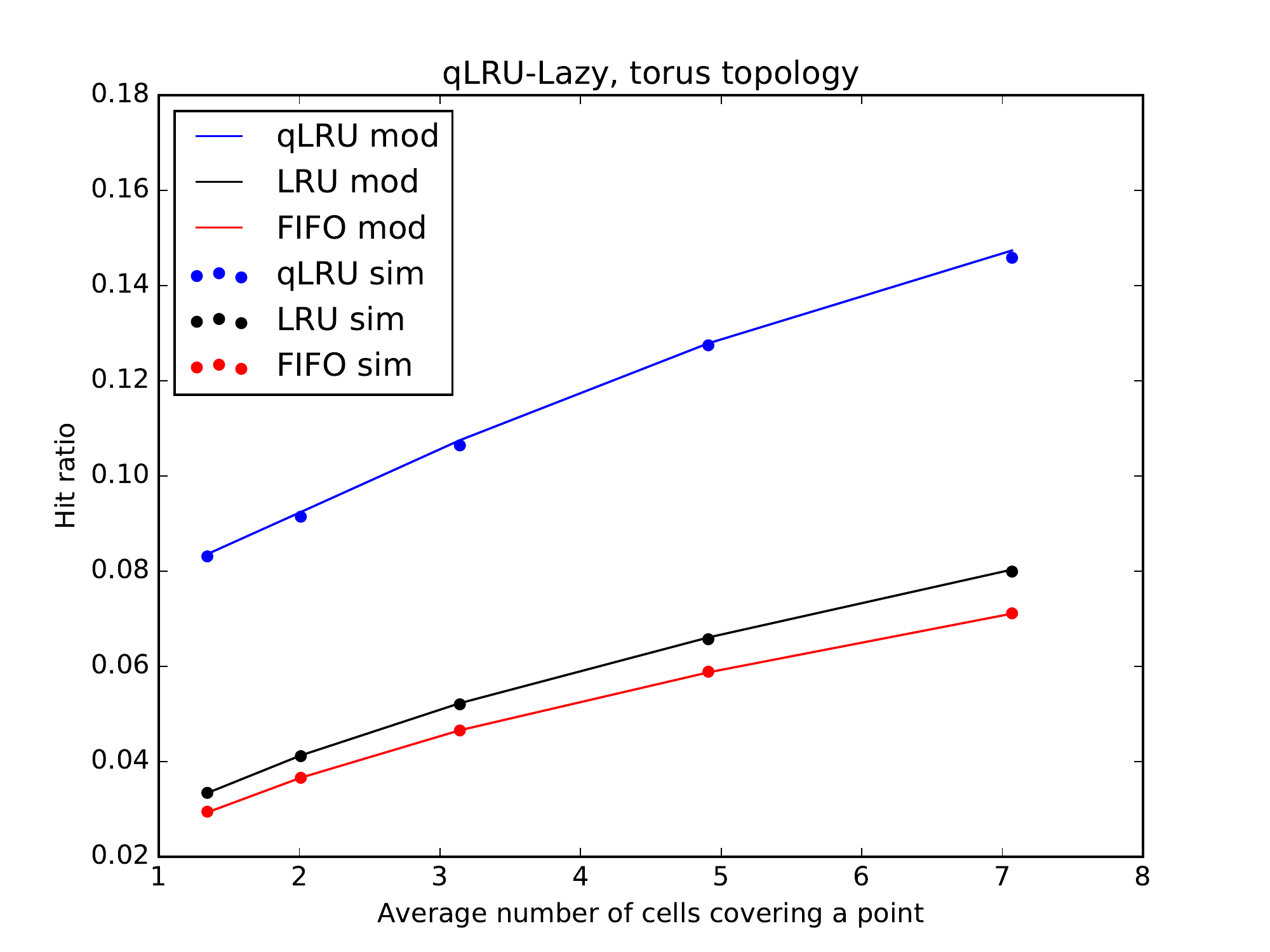}}\\
         \caption{\changes{ Comparison between model predictions and simulations; torus topology \changes{with 9 cells}; $C=100$; IRM traffic model with $\alpha=0.8$; \qlru{} employs $q=0.01$.}}
         \label{f:torus_val}
       \end{figure*}


%
%

\begin{figure}[htbp]
   \centering
   \includegraphics[width=\myFigureScale\linewidth]{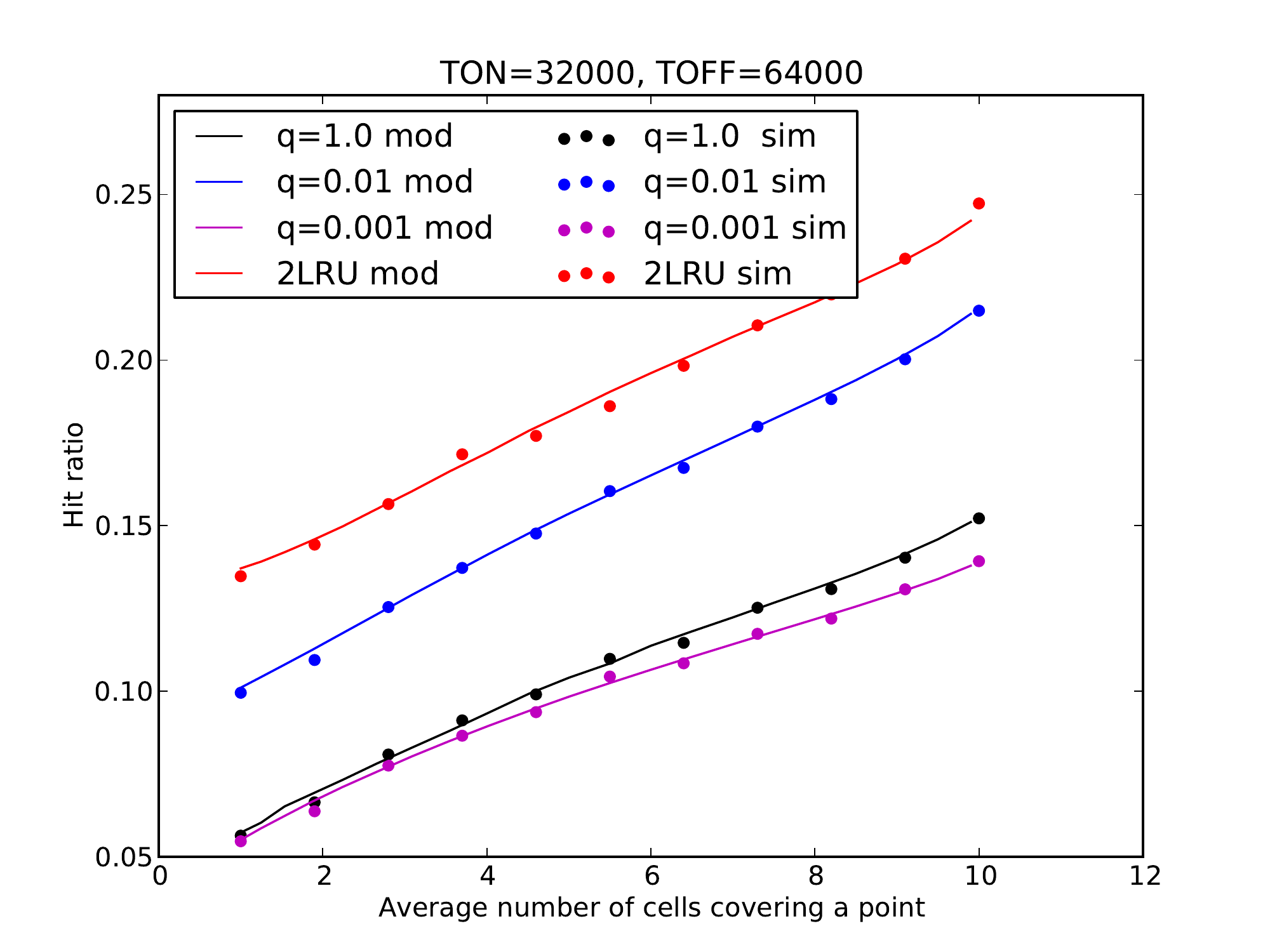}
   \caption{Trefoil: results for the \emph{lazy} update rule and an ON-OFF request process with $\E[T_f^{ON}]=3.2*10^4$ and $\E[T_f^{OFF}]=6.4*10^4$. The cell request rate for the most popular content is  $\Lambda=1.3$.} 
   \label{f:ton16000}
\end{figure}

Finally, Fig.~\ref{f:ton16000} shows that the model is also accurate when the request process differs from IRM. The curves have been obtained for a trefoil topology under the \emph{lazy} update rule and the ON-OFF traffic model described and studied in Sec.~\ref{s:temporal_locality}. 
In the figure we also show some results for \twolrulazy. As \qlru, upon a miss, \twolru{} prefilters the contents to be stored in the cache. \qlru{} does it probabilistically, while \twolru{} exploits another \lru{} cache for the metadata. Under IRM, their performance are qualitatively similar, but \twolru{} is known to be more reactive and then 
better performing when the request process exhibits significant temporal locality~\cite{garetto16}. Our results in Fig.~\ref{f:ton16000} confirm this finding. In particular, as $q$ decreases, the performance of \qlru{} first improves (compare $q=0.01$ with $q=1$) because \qlru's probabilistic admission rule filters the unpopular content, and then worsens (see the curve for $q=0.001$) when \qlru{} dynamics' timescale becomes comparable to $T_{ON}$.

%
%


\section{The Lazy Update Rule}
\label{s:lazyqlru}

Even if  the focus of the previous section has mainly been on the validation of our model, the reader may  have observed by  
looking at Figures~\ref{f:flower_val} and~\ref{f:torus_val} that the update rule \lazy{} 
performs significantly better than \one{} and
(to a lesser extent)  \blind{}, especially when the cellular network is particularly dense. This improvement comes at the cost of a minimal communication overhead: an additional bit to be piggybacked into  every user's request to indicate whether  the content is available at some of the other cells  covering the user. In this section we use our model to further investigate the performance of the \lazy{} update rule.

\begin{figure}[htbp]
   \centering
   \includegraphics[width=\myFigureScale\linewidth]{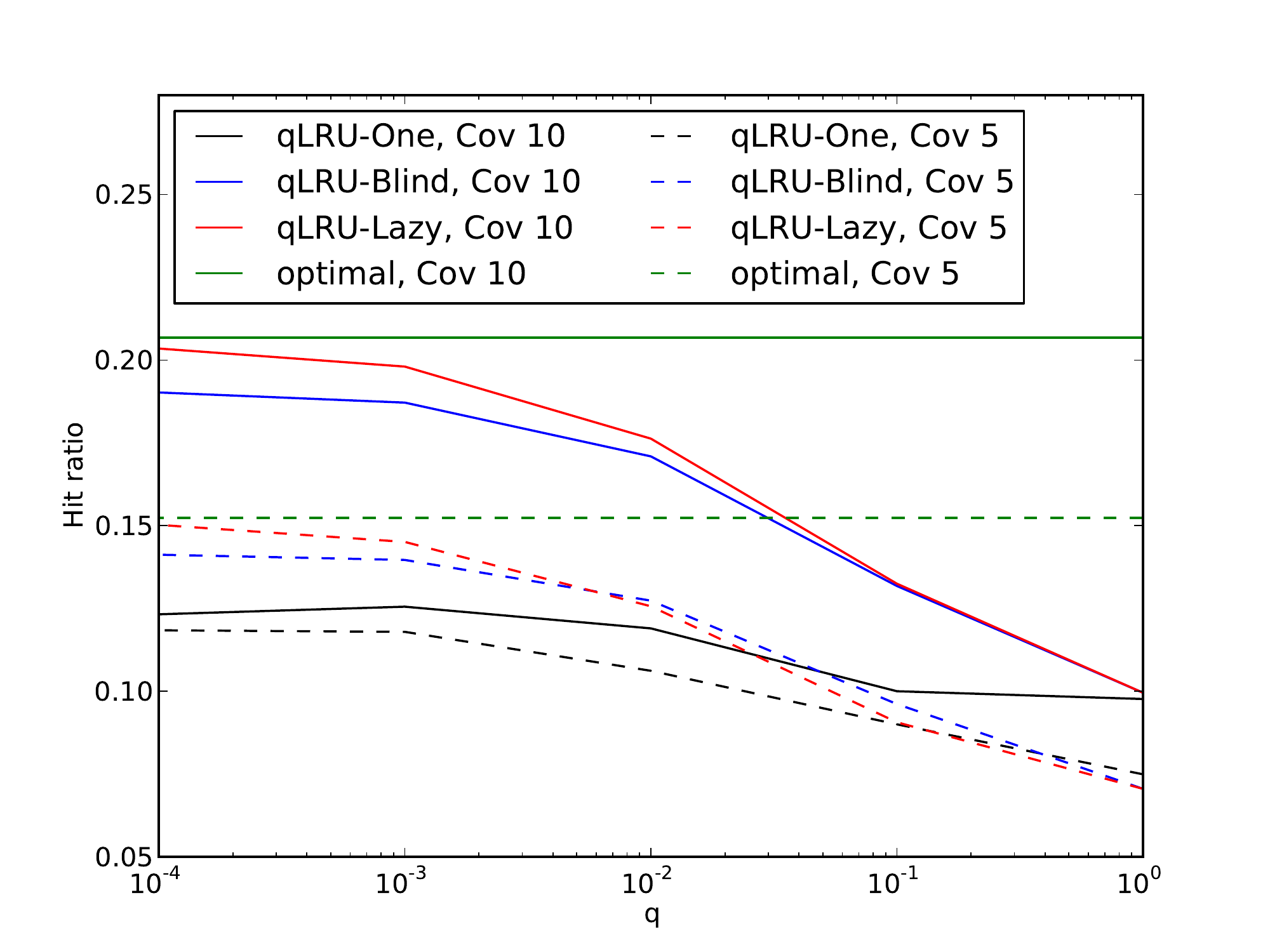}
   \caption{Performance of \qlru{} coupled with different update rules in a trefoil network (model results).} 
   \label{f:qlru_comparison}
\vspace{-5 mm}
   \end{figure}

First, we present  in Fig.~\ref{f:qlru_comparison} some results for \qlru{} coupled with the different update rules. The curves show the hit ratio versus the parameter $q$ achieved by the different policies.
The topology is a trefoil with 10 cells. Results are reported for two values of cell overlap, corresponding  to the cases where a user is  covered on average by $5$ and $10$ BSs. As a reference, also the optimal achievable hit ratio is shown by the two horizontal green lines \changes{(in this particular scenario, the optimal allocation can be obtained by applying the greedy algorithm described below in Sec.~\ref{s:greedy})}. \qlrulazy{} significantly outperforms \qlruone{} for  small values of $q$, with relative gain that can be as high as $25$\% for the $5$-coverage and $65$\% for the $10$-coverage. The improvement with respect  to \blind{} is smaller, but what is remarkable is that \qlrulazy{} appears to be able to asymptotically approach  the performance of the optimal allocation. In the following we will prove that i) this is indeed the case for the trefoil topology and ii) \qlrulazy{} achieves a locally optimal allocation in a general scenario. For a single cache, it has already been proven that \qlru{} asymptotically maximizes the hit ratio when $q$ converges to $0$ (see~\cite{garetto16} for the case of uniform content size contents and~\cite{neglia16itc} for the case of heterogeneous size), but, to the best of our knowledge,  no optimality results are available for a multi-cache scenario as the one we are considering.
Before proving optimality, we discuss what is the optimal allocation and we provide some intuitive explanation about \lazy{} good \changes{performance}.

%
%
%

\subsection{Optimal content allocation and a new point of view on \lazy}
\label{s:greedy}
If content popularities are known and stationary, one can allocate, once and for all, contents to caches in order to maximize the global hit ratio. 
Formally, the following integer maximization problem can be defined: 
\begin{align}
\label{e:static_opt}
& \text{maximize}
& & \sum_{f=1}^F  \lambda_f \mu\Bigg(\bigcup_{\begin{subarray}{c} b |   x_f^{(b)}=1 \end{subarray}} S_b\Bigg) \\ 
& \text{subject to}
& & \sum_{f=1}^F x_f^{(b)} =C  \;\;\; \forall b = 1, \ldots B,\nonumber\\ 
& & & x_f^{(b)} \in \{0,1\} \;\;\; \forall f =1, \ldots F, \;\;\; \forall b = 1, \ldots B.\nonumber 
\end{align}

%
Carrying on an analysis similar to that in  \cite{shanmugam13}, it is possible to show that this problem i) is NP-hard (e.g. through a reduction to the 2-Disjoint Set Cover Problem), ii) can be formulated as the maximization of a monotone sub-modular set function with matroid constraints. It follows that the associated greedy algorithm provides a $1/2$-approximation for problem~\eqref{e:static_opt}.

Let us consider how the greedy algorithm operates. Let $\X(l-1)\in\{0,1\}^{B\times F}$ describe the allocation at the \mbox{$(l-1)$-th} step of the greedy algorithm, i.e.~the matrix element \mbox{$(\X(l-1))_{f,b}=x_f^{(b)}(l-1)$} indicates if at step $l-1$ the algorithm places content $f$ at cache $b$. At step $l$, the greedy algorithm  computes for each content $f$ and each cache $b$ the marginal improvement for the global hit ratio to store a copy of $f$ at cache $b$, given the current allocation $\X(l-1)$, that is
\begin{equation}
\label{e:marginal}
\lambda_f \mu\Bigg(S_b \setminus \bigcup_{b' | x_f^{(b')}(l-1)=1} S_{b'}\Bigg)
\end{equation}
The pair $(f_l , b_l )$ leading to the largest hit ratio increase is then selected and the allocation is updated by setting $x^{(b_l)}_{f_l} = 1$.
The procedure is iterated until all the caches are full.

We observe that \eqref{e:marginal} is exactly the request rate  that drives the dynamics of \qlrulazy{} in state $\X_f(l-1)$,
as indicated in~\eqref{e:lazy_rate}. 
\changes{Upon a miss for content $f$, \qlrulazy{} inserts it with a probability that is proportional to the marginal increase of the global hit ratio provided by adding the additional copy of content $f$. This introduces a stochastic drift toward local maxima of the hit ratio.}
As we said above,  when $q$ vanishes, it is known that an isolated \qlru{} cache tends to  store deterministically 
the top popular contents, then one can expect  each \qlrulazy{} cache to store the contents with the largest marginal request rate given the current allocation at the other caches.
Therefore, it seems licit to conjecture that a system of \qlrulazy{} caches  asymptotically converges at least to a local maximum for the hit ratio (the objective function in~\eqref{e:static_opt}). { Section~\ref{s:qlrulazy_convergence_general} shows that this is indeed the case.} Before moving to that result, we show that in particular \qlrulazy{} achieves the maximum hit ratio in a trefoil topology.

\subsection{In a trefoil topology \qlrulazy{} achieves the global maximum hit ratio}
\label{s:qlrulazy_convergence}
Now, we formalize the previous arguments,
showing that as $q$ tends to $0$, \qlrulazy{} content allocation
  converges to an optimal configuration in which the set of contents maximizing the global hit ratio
is stored at the caches. This result holds for the trefoil topology  under our model. 

We recall that that the trefoil topology exhibits  a complete cell symmetry and that the hit ratio of any allocation is invariant under cell label permutations. A consequence is that the hit ratio depends only on the number of copies of each file that are stored in the network, while it does not depend on where they are stored as far as we avoid to place multiple copies of the same file in the same cache, that is obviously unhelpful. It is possible then to describe a possible solution simply as an $F$-dimensional vector $\mathbf k=(k_1,k_2, \dots k_F)$, where $k_f$ denotes the number of copies of content $f$. Under \qlrulazy{} we denote by $\pi(q,\mathbf k)$, the stationary probability that the system is in a state with allocation $\mathbf k$.

The optimality result follows from combining the two following propositions (whose complete proofs are in \citeapp{a:qlrulazy_optimality}):
\begin{prop}
\label{p:greedy_flower}
In a trefoil topology, an allocation of the greedy algorithm for Problem~\eqref{e:static_opt} is optimal. 
\end{prop}
The proof relies on mapping problem~\eqref{e:static_opt} to a knapsack problem with $F\times C$ objects with unit size for which the greedy algorithm is optimal.

We observe that for generic values of the parameters, all the marginal improvements considered by the greedy algorithm are different and then the greedy algorithm admits a unique possible output (apart from BSs label permutations). 

\begin{prop}
\label{p:qlrulazy_convergence}
Consider a trefoil topology and assume there is unique possible output of the greedy algorithm, denoted as $\mathbf k^*=(k_1^*, k_2^*, \dots k_F^*)$. Then,  under the approximate model in Sec~\ref{s:model}, a system of \qlrulazy{} caches asymptotically converges to $\mathbf k^*$ when $q$ vanishes in the sense that 
\[\lim_{q \to 0} \pi(q,\mathbf k^*)=1.\]
\end{prop}
In order to prove this result, we  write down the explicit stationary probability for the system, taking advantage of the fact that the MC is reversible, and we study its limit. 
\changes{ In conclusion the greedy algorithm and \qlrulazy{} are equivalent in the case of   trefoil topology.}	

%
%
%

\subsection{\qlrulazy{} achieves  a local maximum hit ratio}
\label{s:qlrulazy_convergence_general}
We say that a caching configuration $\mathcal{C}$ is locally optimal if it  provides the highest aggregate hit rate among all the caching 
 configurations which can be obtained  from $\mathcal{C}$ by replacing one content in one of the caches.
  

 \begin{prop}
\label{p:qlrulazy_convergence_general}
A spatial network of \qlrulazy{} caches asymptotically achieves a locally-optimal caching configuration when  $q$ vanishes.\footnote{In the most general case, the adoption of different parameters $q$ is required at different cells for the implementation of the 
\qlru\ policy.} 
 \end{prop}
The proof is in \citeapp{a:qlrulazy_convergence_general}.
 In this general case  the difficulty of proving the assertion stems from the fact that the MC representing content dynamics  is not anymore reversible, 
 and it is then difficult to derive an analytical expression for its steady state  distribution. Instead, our proof relies on results for regular perturbations of Markov chains~\cite{young93}.

The analytical results in this section justify why for small, but strictly positive, values of $q$, \qlrulazy{} performs better than \qlrublind{} and \qlruone. 
More in general, what seems fundamental to approach the maximum hit ratio is the coupling of the \lazy{} update rule, that reacts to the ``right marginal benefit'' for problem~\eqref{e:static_opt}, with a caching policy that is effective to store the most popular contents. \qlru{} is one of them, \twolru{} is another option. Moreover, \twolru{} has been shown to react faster to popularity changes. For this reason, in the next section we also include results for \twolrulazy{}.
\changes{At last we wish to remark that the (static) cache configuration selected by greedy algorithm  is in general not locally optimal, as a consequence  of  the greedy nature of the algorithm and the fact that marginal gains at a cell change during the execution of the algorithm (since they depend 
on the configuration of neighbouring cells).}
 

\section{Performance in a Realistic Deployment}
\label{s:realistic}

In this section we  evaluate the performance of the \lazy{} update rule in a more realistic scenario. To this purpose, we have extracted the positions of $10$ T-Mobile BSs in Berlin from the dataset in~\cite{bs_dataset} and we use a real content request trace from Akamai Content Delivery Network~\cite{neglia17tompecs}. The actual identity of the users and of the requested objects was obfuscated. The BS locations are indicated in Fig.~\ref{f:berlin}. We refer to this topology simply as the Berlin topology.  The trace includes 400 million requests issued over 5 days from users in the same geographical zone for a total of 13 million unique contents. In our simulations we randomly assign the requests to the users who are uniformly spread over the area.

	\begin{table}
	   \centering
	    \caption{Trace: basic information\label{table:basic_stats}}{
	    \centering
	    \begin{tabular}{|l|r|}
	      \hline
	      Time span	& $5$ days\\
	      \hline
	      Number of requests received  & $4 \cdot 10^8$ \\
	      \hline
	      Number of distinct objects  & $13  \cdot 10^6$\\    
	      \hline
	    \end{tabular}}
	  \end{table}

\begin{figure}[htbp]
   \centering
   \includegraphics[width=\myFigureScale\linewidth]{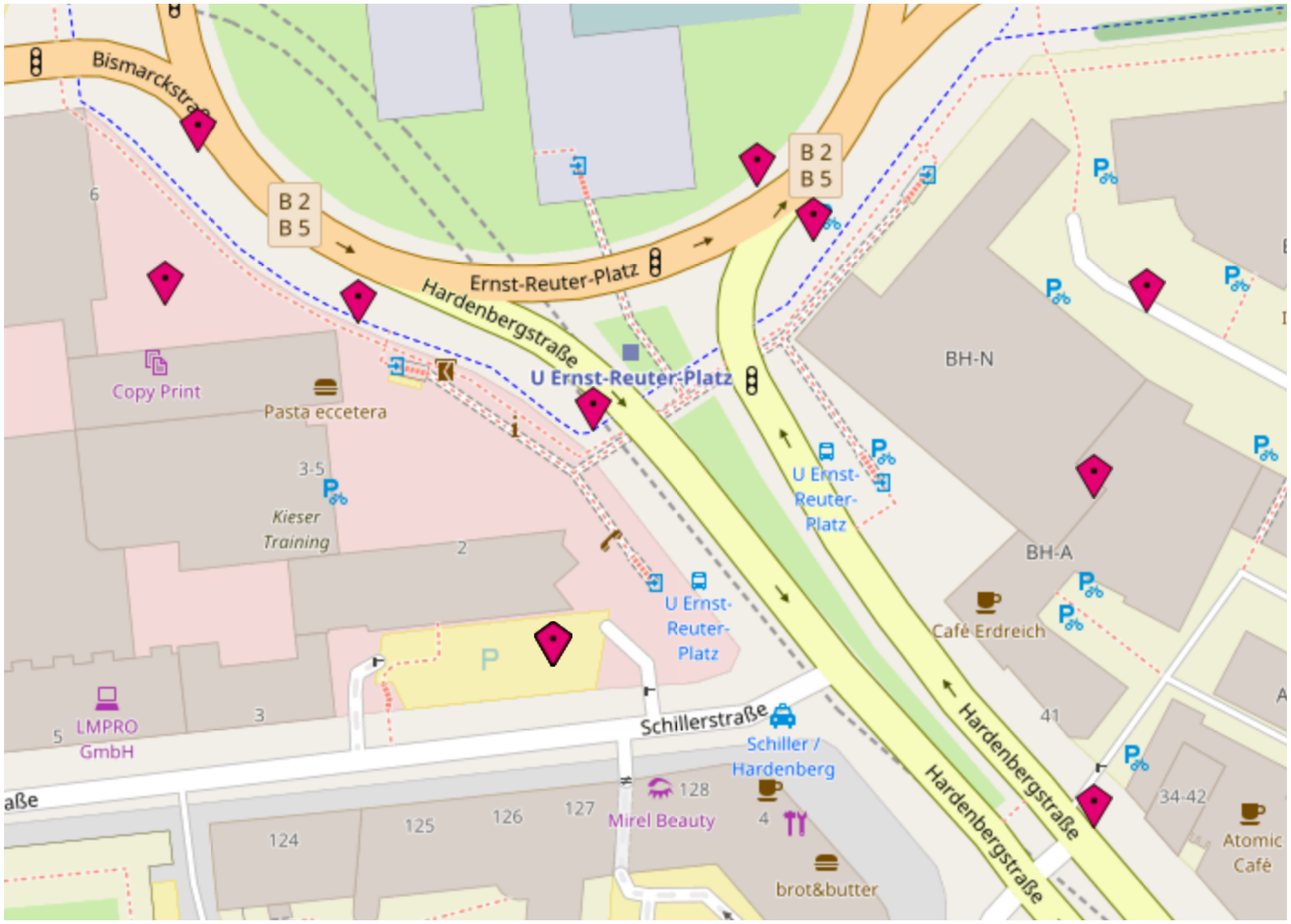}
   \caption{T-Mobile BS configuration in Berlin.} 
   \label{f:berlin}
\end{figure}

Figure~\ref{f:berlin_comparison} compares the performance of different caching policies in this scenario, when the transmission range of the BSs varies from 25 to 250 meters and correspondingly a user is covered \changes{on average} by 1.1 up to 9.4 BSs. We observe 
that the \lazy{} update rule still outperforms \one{} and \blind{} when coupled with \qlru{} or \twolru. 
Moreover, for the higher density scenarios, \twolrulazy, \twolrublind, \qlrulazy{} and (to a minor extent) \qlrublind{}   outperform the  static allocation that has been obtained by the greedy algorithm assuming known the request rates  of each content over the future $5$ days. While we recall that the greedy algorithm provides only  a $1/2$-approximation 
of the optimal allocation  for problem~\eqref{e:static_opt}, 
we highlight that this apparently surprising result is most likely to be due to the non-stationarity of the request process. In this case an uninformed dynamic policy (like \qlru{} or \twolru) can outperform an informed static one, by dynamically adapting content allocation in caches to the short-term request rate of contents.

\begin{figure}[htbp]
   \centering
   \includegraphics[width=\myFigureScale\linewidth]{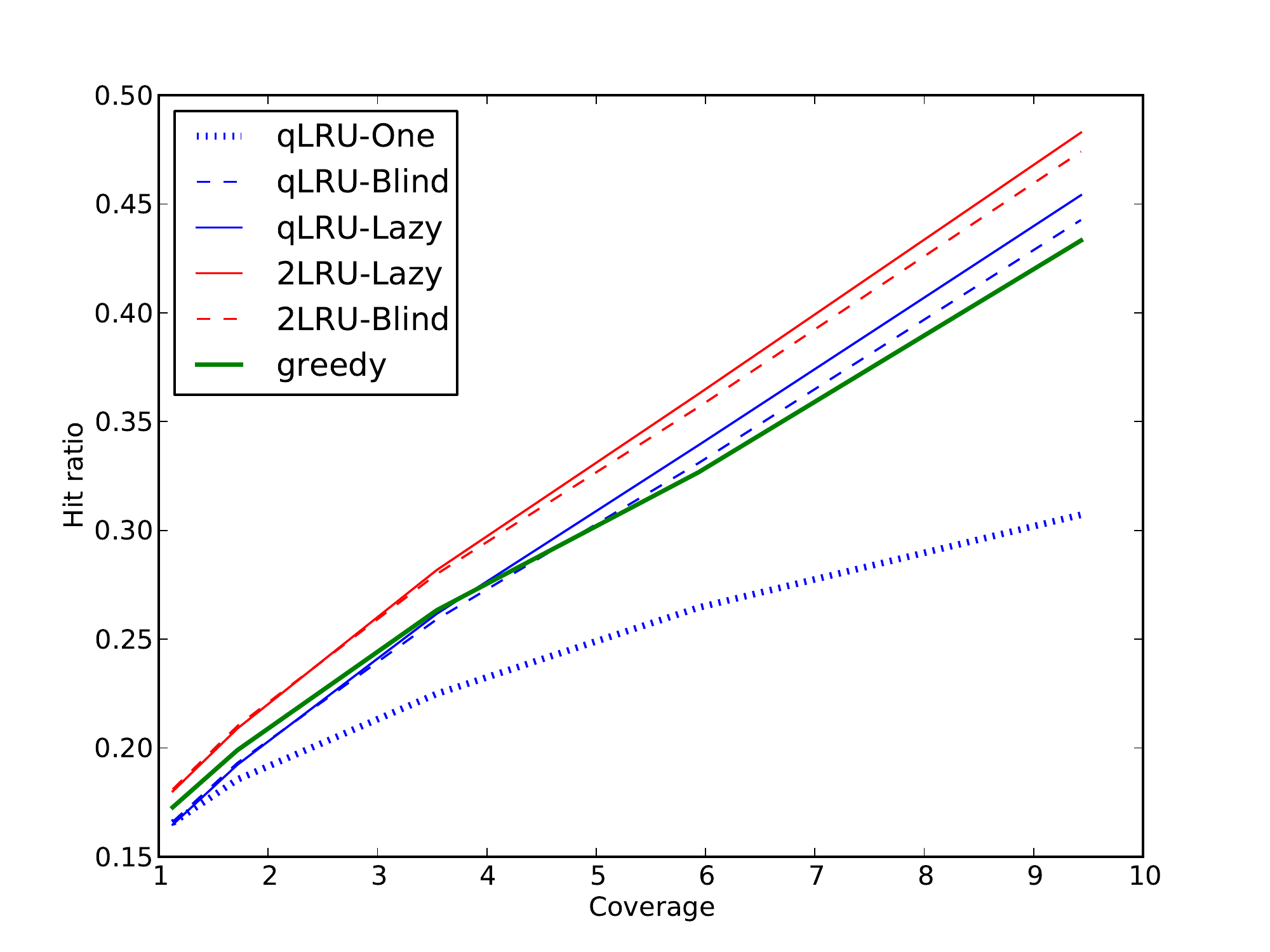}
   \caption{Berlin topology and real CDN request trace. \qlru{} employs $q=0.01$.} 
   \label{f:berlin_comparison}
\end{figure}

In order to deepen the comparison between uninformed and informed policies, we have considered the  operation scenario that is usually suggested from supporters of informed policies (see e.g.~\cite{shanmugam13}): the optimal allocation is computed the day ahead and contents are pushed to the caches during the night when the network is unloaded. Figure~\ref{f:comparison_day_by_day} shows then the performance  for two coverage values ($1.1$ and $5.9$) on a daily basis as well as for the whole 5 days. In this case, the oracle greedy algorithm computes a static content allocation each night knowing exactly the future request rates for the following day. Instead, the forecast greedy algorithm uses the request rates seen during the current day as an estimation for the following one.
The oracle greedy benefits from the knowledge of the request rates over a single day: it can now correctly identify and store those  contents that are going to be popular the following day, but are not so popular over the whole trace. For this reason,  it outperforms the greedy scheme that receives as input  the average   rates over the whole 5-day period. When cells have limited overlap (Fig.~\ref{f:comparison_day_by_day}~(a)), the oracle greedy algorithm still outperforms the dynamic policies, but \twolrulazy{} is very close to it. Interestingly, in the higher density setting (Fig.~\ref{f:comparison_day_by_day} (b)), this advantage disappears.  The performance of the \twolrulazy{} allocation  becomes preferable than both oracle greedy (with daily rates), and \qlrulazy{}. 
Temporal locality appears to have a larger impact in  high density scenarios!

At last we wish to remark that  the allocation of the oracle greedy algorithm is an ideal one, because it assumes the future request rates to be known. A practical algorithm will be   necessarily based on some estimates such as the forecast greedy. Our results show a significant performance loss due to this incorrect input (as already observed in~\cite{neglia18ton}). Both \twolrulazy{} and \qlrulazy{} perform significantly better  than forecast greedy (\twolrulazy{} guarantees  between 10\% and 20\% improvement).
\changes{Interestingly, our results contradict one of the conclusions in~\cite{elayoubi15}, i.e.~that at the BS level reactive caching policies would not be efficient because the content request rate is too low, and content prefetching would perform better. We observe that \cite{elayoubi15} considers a single BS scenario and that perfect popularity knowledge is available. We have performed some additional simulations considering the current typical request rate at a BS as identified in~\cite{elayoubi15} and we still observe qualitatively the same behaviour illustrated in Fig.~\ref{f:comparison_day_by_day}. These additional experiments are described in Appendix~\ref{a:roberts_experiments}. Moreover, data traffic rate in cellular networks is constantly increasing and this improves the performance of reactive policies (but not of prefetching) as already observed in~\cite{elayoubi15}.}

\begin{figure}%
         \centering
         \subfloat[Coverage 1.1]{\includegraphics[width=\myFigureScale\linewidth]{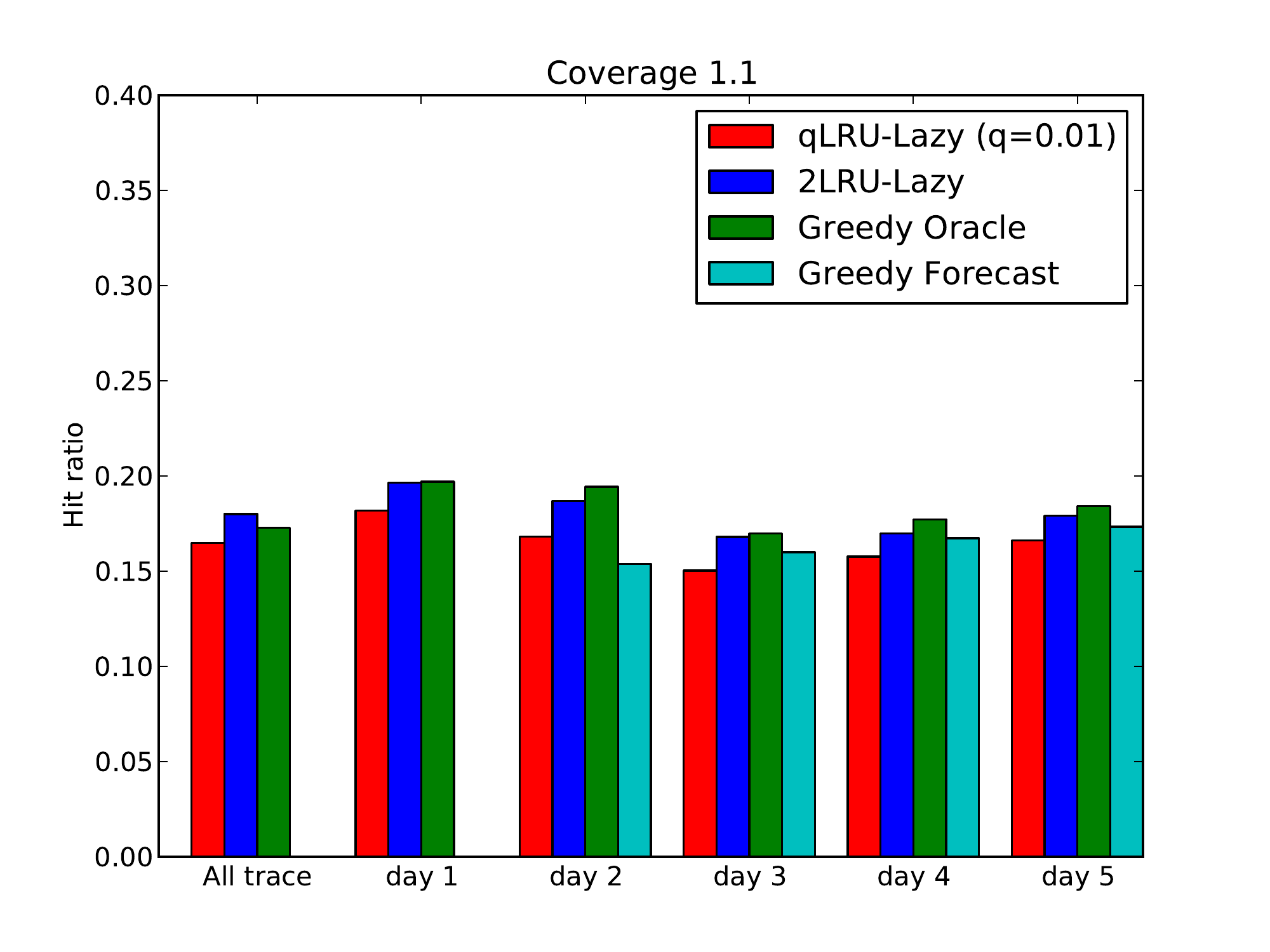}}\hspace{1cm}
         \subfloat[Coverage 5.9]{\includegraphics[width=\myFigureScale\linewidth]{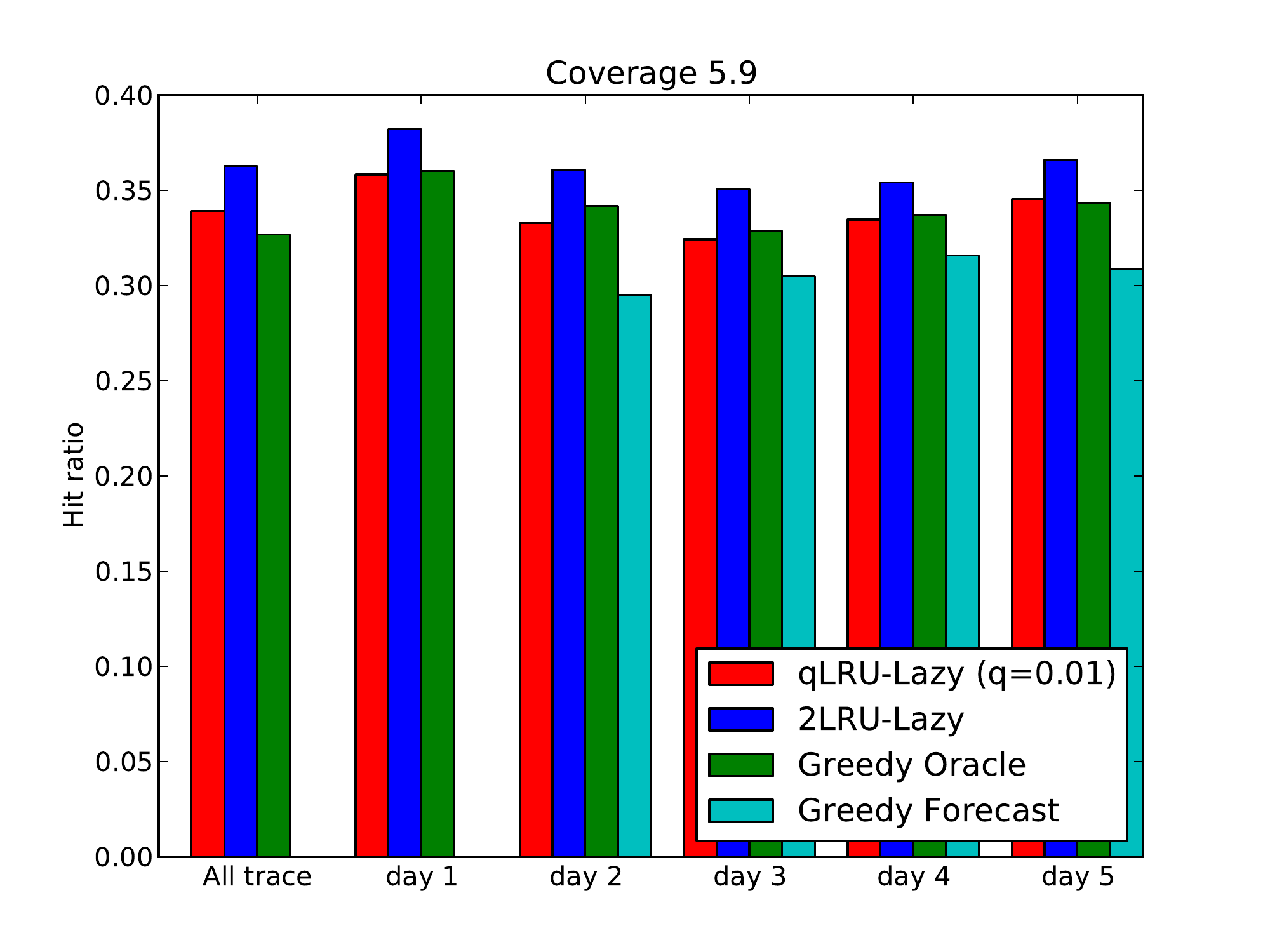}}\\
         \caption{Berlin topology and real CDN request trace. Comparison of the different policies over the whole trace and for each of the $5$ days. \qlru{} employs $q=0.01$.}
         \label{f:comparison_day_by_day}
       \end{figure}

\section{Conclusions} 
\label{s:conclusions}
In this paper, we  have shown that ``uniformed'' schemes can effectively implicitly coordinate different caches 
in dense  cellular systems when smart (but simple) update policies like \lazy{} are used.  Indeed we show that they can  achieve
a performance, which is comparable to that of the best ``informed''
schemes in static traffic  scenarios. Moreover, ``uniformed'' schemes  better adapt  to  dynamic   scenarios,
often outperforming implementable ``informed'' schemes.
For once, then, sloth is not the key to poverty, not at least to poor performance.

We have also proposed a new approximate analytical framework to assess the performance of \lq \lq uniformed'' schemes. The predictions  of our model are extremely accurate (hardly distinguishable from Monte Carlo simulations in most cases). 

This work was partly funded by the French Government (National Research Agency, ANR) through the ``Investments for the Future'' Program reference \#ANR-11-LABX-0031-01.


%
\IEEEpeerreviewmaketitle




\bibliographystyle{IEEEtran}
\bibliography{caching.bib}



\appendices
 \section{Insensitivity for \fifo{} caches}
\label{a:fifo_insensitivity}
 
The proof of Proposition~\ref{p:insensitivity} follows.
 
\begin{proof}
The vector $\X_f(t)$ indicates where copies of content $f$ are stored at time $t$. Under CTA for \fifo, when a copy is inserted at cache $b$, a timer is set to the deterministic value $T_{c}^{(b)}$ and decreased over time. When the timer reaches zero, the content is erased from cache $b$. We denote by $c_f^{(b)}(t)$ the residual value of such timer at time $t$. The system state at time $t$, as regards content $f$, is then characterized by $\mathbf Y_f(t) \triangleq (\X_f(t), \{c_f^{(b)}(t), \forall b \; | X_f^{(b)}(t)=1\})$, i.e.~by the current allocation of content $f$ copies and their residual timers. The process $\mathbf Y_f(t)$ is a Generalized Semi-Markov Process (GSMP)~\cite{konstantopoulos89}.

Let $\pi(\x_f)$ be the stationary distribution of $\X_f(t)$. This distribution is in general a function of the timer distributions. If it depends on them only through their expected values, then the GSMP is said to be \emph{insensitive}. In this case the stationary distribution remains unchanged if we replace all the timers with exponential random variables and then the GSMP simply becomes a continuous time Markov Chain with state $\X_f(t)$. Then, Approximation A$1$ is correct whenever the GSMP $\mathbf Y_f(t)$ is insensitive.

A GSMP is insensitive if and only if its stationary distribution $\pi(\x_f)$ satisfies some partial balance equations (a.k.a.~Matthes' conditions)~\cite[Eq.~(2.2)]{konstantopoulos89}, as well as the usual global balance equations. For our system, Matthes' conditions can be written as 
\begin{equation}
\label{e:reversibility}
\pi\left(0,\x_f^{(-b)}\right) \Lambda_f^{(b)} = \pi\left(1,\x_f^{(-b)}\right) \frac{1}{T_c^{(b)}} \;\; \forall \x_f^{(-b)},
\end{equation}
i.e.~the rate at which the timer $c_f^{b}$ is activated is equal to the rate at which it expires. Conditions \eqref{e:reversibility} are equivalent to the corresponding MC being reversible. This completes the proof.

\end{proof}

\section{\qlrulazy's optimality}
\label{a:qlrulazy_optimality}

In the trefoil topology, any cell is equivalent to any other, so it does not matter where the copies of a given content are located, but just how many of them there are. 
Let $\Lambda_f(k)$ be the total request rate for content $f$ from users located in $k$ cells,\footnote{
	It does not matter which ones, because of the symmetry.
} i.e.:
\[\Lambda_f(k)\triangleq\lambda_f \mu\Bigg( \bigcup_{b=1}^{k} S_b\Bigg).\]
$\Lambda_f(k)$ corresponds then to content $f$ hit ratio when $k$ copies of the content are stored at the caches.
Moreover, we denote by $\Delta \Lambda_f(k)\triangleq \Lambda_f(k) - \Lambda_f(k-1)$ the marginal increase of the hit ratio due to adding a $k$-th copy of content $f$.
We observe that $\Delta \Lambda_f(k)$ is decreasing in $k$ and that it holds: 
\[\Lambda_f(k)=\sum_{h=1}^k \Delta \Lambda_f(h).\]

\subsection{Proof of Proposition~\ref{p:greedy_flower}}
\label{s:proof_greedy}
\begin{proof}
The proof is rather immediate. 
First observe that
by exploiting the properties of the cell-trefoil topology, \eqref{e:static_opt}  can be rewritten as:

\begin{align}
\label{e:static_opt_semp}
\text{maximize}
 & \sum_{f=1}^F  \Lambda_f (k_f), \\ 
 \text{subject to}
  & \sum_{f=1}^F k_f =B\times C . \nonumber
\end{align}

Now, caching problem \eqref{e:static_opt_semp} can be easily mapped  to a (trivial)  knapsack problem with $F\times C$
objects of unitary size, according to the following lines: for every content $f$  we define $C$  different virtual objects $(f,h)$ with $1 \le h \le C$, with associated weights:
\[
 w_{(f,h)}=\Delta\Lambda_f(h),
\]
i.e.~the weight $w_{(f,h)}$ is equal to the marginal increase of the hit ratio, which is obtained by storing the 
$h$-th  copy of content $f$ into the caching system.

The objective 
of the  knapsack problem  is to find the set   ${\mathcal S}_{\text{opt}}$ of $F\times C$ objects,  which maximizes 
the sum of all the associated weights.  Indeed \eqref{e:static_opt_semp}  can be rewritten as: $\max_{(k_1,\ldots, k_f, \ldots, k_F)}
\sum_f\sum_{h=1}^{k_f}  w_{(f,h)}$. In particular, observe that  since $ w_{(f,h)}\le  w_{(f,h-1)}$, virtual  object $(f,h)\in  {\mathcal S}_{\text{opt}}$ only if 
$(f,h-1)\in  {\mathcal S}_{\text{opt}}$.  This implies that ${\mathcal S}_{\text{opt}}$ provides a feasible  solution for the 
original caching problem,  where $k_f$ is equal to the largest $h$ such that $(f,h)\in {\mathcal S}_{\text{opt}}$.

Finally, note that, by construction,  ${\mathcal S}_{\text{opt}}$ is the set composed of the $F\times C$ objects  with the largest 
value; therefore by construction,  ${\mathcal S}_{\text{opt}}$  corresponds to:
i)    the caching allocation that maximizes the global hit rate (i.e. the 
allocation that solves \eqref{e:static_opt_semp});
ii)  moreover, it is the only solution of the greedy algorithm, under the  assumption that   object values are all different, (i.e.~for generic values of the parameters).

\end{proof}

\subsection{Proof of Proposition~\ref{p:qlrulazy_convergence}}
\begin{proof}
Under our approximated model,  system dynamics are described by $F$ Markov chains, one for each content, coupled by the characteristic times.
The symmetry of the trefoil topology implies that the characteristic time at each cache has the same value that we denote simply as $T_c$.

Every MC is a birth-death process. In particular, under \qlrulazy,  for content $f$, the transition rate from state $k_f-1$ to $k_f$ is 
\[r(k_f-1,k_f) \triangleq q \left(\Lambda_{f}(B) - \Lambda_f(k_f-1) \right),\] 
and from state $k_f$ to $k_f-1$ it is 
\[r(k_f,k_f-1) \triangleq k_f \frac{\Delta \Lambda_f(k_f)}{e^{\Delta\Lambda_f(k_f)  T_c }-1}.\]

Let us define $\rho_f(k) \triangleq r(k_f-1,k_f)/r(k_f,k_f-1)$. The stationary probability to have $k_f$ copies of content $f$ is then
\begin{align*}
	\pi_f(k_f) & = \frac{\prod_{h=1}^{k_f} \rho_f(h)}{1+\sum_{k=1}^B \prod_{h=1}^k \rho_f(h)}\\
			& = \frac{A_{f,k_f} \prod_{h=1}^{k_f} q \left( e^{\Delta\Lambda_f(h) T_c} -1\right)}{1 + \sum_{k=1}^B A_{f,k} \prod_{h=1}^k q \left( e^{\Delta\Lambda_f(h) T_c} -1\right)}, 
\end{align*}
where 
\[A_{f,k} \triangleq \prod_{h=1}^k \frac{\Lambda_{f}(B) - \Lambda_f(h-1)}{h \Delta\Lambda_f(h)}\]
are values that do not depend on $q$ or $T_c$ and they will not play a role in the following study of the asymptotic behaviour.

Under CTA, the buffer constraint is expressed imposing that the expected number of contents at a cache is equal to the buffer size.
If the system is in state $k_f$, any given BS has probability $k_f/B$ to be one of the $k_f$ storing it, then the buffer constraint is
\[
\sum_{f=1}^F \sum_{k=1}^B \frac{k}{B} \pi_f(k)=C,
\]
or equivalently:
\begin{equation}
\label{e:constraint2}
\sum_{f=1}^F \sum_{k=1}^B k \pi_f(k)=C \times B.
\end{equation}
We focus now our attention on the stationary distribution when $q$ converges to $0$. As $q$ changes, the characteristic time changes as well. We write $T_c(q)$ to express such dependence. When $q$ converges to $0$, $T_c(q)$ diverges, otherwise all the probabilities $\pi_f(k_f)$ would converge to $0$ and constraint~\eqref{e:constraint2} would not be satisfied. It follows that:

\begin{align}
\label{e:asymp_prob}
\pi_f(q,k_f) 	& \underset{q\to 0}{\sim} \frac{A_{f,k_f} \prod_{h=1}^{k_f}
\left(q  e^{\Delta\Lambda_f(h) T_c(q)}\right)}{1+ \sum_{k=1}^B A_{f,k} \prod_{h=1}^{k} \left(q  e^{\Delta\Lambda_f(h) T_c(q)}\right)}.
\end{align}


Let us consider a sequence $(q_n)_{n \in \mathbb N}$ that converges to $0$ such that  it exists $\lim_{n \to \infty} \ln(1/q_n)/T_c(q_n)=\hat \Lambda$. The value $\hat \Lambda$ is also said to be a cluster value for the function $\ln(1/q)/T_c(q)$. It holds that for any marginal hit ratio $\Delta \Lambda_f(h)$
\[\lim_{n \to \infty} q_n e^{\Delta\Lambda_f(h) T_c(q_n)}=
\begin{cases}
		+\infty, & \textrm{ if } \Delta\Lambda_f(h) > \hat \Lambda,\\
		0, & \textrm{ if } \Delta\Lambda_f(h) < \hat \Lambda.
\end{cases}		
\]
Let $\hat k_f \triangleq \argmax_h \{\Delta\Lambda_f(h) > \hat\Lambda \}$. 
For generic values of the parameters the values $\{\Delta \Lambda_f(h), \forall f=1,\dots F, h=1,\dots B\}$ are all distinct by hypothesis, then there can be at most one content $f_0$ 
for which it holds $\Delta\Lambda_f(\hat k_{f_0}+1)= \hat\Lambda$.
For any other content it follows that the dominant term in the denominator of \eqref{e:asymp_prob} is $\prod_{h=1}^{\hat k_f} \left(q  e^{\Delta\Lambda_f(h) T_c(q)}\right)$ and then for $f\neq f_0$:
\[\lim_{n \to \infty} \pi_f(q_n,k_f)=
\begin{cases}
		1, & \textrm{if } k_f=\hat k_f,\\
		0, & \textrm{otherwise},
\end{cases}		
\]
i.e. asymptotically exactly $\hat k_f$ copies of content $f$ would be stored.
For content $f_0$, both the term $\prod_{h=1}^{\hat k_{f_0}} \left(q  e^{\Delta\Lambda_f(h) T_c(q)}\right)$ and $\prod_{h=1}^{\hat k_{f_0}+1} \left(q  e^{\Delta\Lambda_f(h) T_c(q)}\right)$ could be dominant. Then all the $\pi_{f_0}(q_n,k_f)$ converge to $0$ for $k_f \notin \{\hat k_{f_0}, \hat k_{f_0}+1\}$. The total expected number of copies stored in the system would then be:
\begin{align*}
\sum_{f \neq f_0} \hat k_f & + \hat  k_{f_0} \pi_{f_0}(0,\hat k_{f_0})+ (\hat k_{f_0}+1) \pi_{f_0}(0,\hat k_{f_0}+1)  \\
				& = \sum_{f \neq f_0} \hat k_f  +  \hat k_{f_0} +\pi_{f_0}(0,\hat k_{f_0}+1). 
\end{align*}
Because of~\eqref{e:constraint2}, this sum has to be equal to the integer $C\times B$, then one of the two following mutually exclusive  possibilities must hold, or
\[ \sum_{f \neq f_0} \hat k_f  +  \hat k_{f_0} = C \times B
\]
and then  $\pi_{f_0}(0,\hat k_{f_0}+1)=0$, or
\[ \sum_{f \neq f_0} \hat k_f  +  \hat k_{f_0} + 1 = C \times B
\]
and then  $\pi_{f_0}(0,\hat k_{f_0}+1)=1$. In any case, the conclusion is that, when $q$ converges to $0$, for each content $f$ a fixed number of copies $k_f$ is stored at the cache. $k_f$ is such that the marginal hit-ratio increase due to the $k_f$-th copy is among the largest $C \times B$ marginal hit-ratios (and the $(k_f+1)$-th copy is not among them). This allocation coincides with the solution of the greedy algorithm. 

%
%
%
%
%

\end{proof}

 \subsection{Proof of Proposition~\ref{p:qlrulazy_convergence_general}}
 \label{a:qlrulazy_convergence_general}
 \begin{proof}

 For a given content $f$, let  $\bold{x}_f$ and $\bold{y}_f$  be two possible states of the MC. We say that 
  $\bold{x}_f \le \bold{y}_f$    whenever  $x_f^{(b)} \le  y_f^{(b)}$ for each $b$; furthermore we denote with  
 $|\bold{x}_f|=\sum_b x_f^{(b)}$ the number of stored  copies of the content in state $\bold{x}_f$, which we call weight of the state $\bold{x}_f$.

 Now  observe that by construction, transition rates in the MC are different from 0 only between pair of states  $\bold{x}_f$ and $ \bold{y}_f$,  
   such that: i)   $\bold{x}_f \le \bold{y}_f$ , ii)  $|\bold{x}_f |= |\bold{y}_f |-1$. 
   In such a case we say that $\bold{y}_f$ is a parent of $\bold{x}_f$ and   $\bold{x}_f$ is a son of $\bold{y}_f$. Moreover we say that $\x_f \to \y_f$ is an \emph{upward} transition, while  $\y_f \to \x_f$ is a downward transition.

   Let $\bold{y}_f$ be parent of $\bold{x}_f$ and let   $b_0$ be the index such that    $x_f^{(b_0)}<y_f^{(b_0)}$,  we have that the upward rate $\rho_{[\bold{x}_f \to \bold{y}_f]} = q \Lambda_f^{(b_0)}(\bold{x}_f)=\Theta(q)$ and the downward rate 
  $\rho_{[\bold{y}_f \to \bold{x}_f]} = 	\frac{ \Lambda_f^{(b_0)}(\bold{x}_f) }{e^{  \Lambda_f^{(b_0)}(\bold{x}_f) T_C^{(b_0)} }-1}$.
  
Now, as $q\to 0$  for every $f$ every upward rate  $r_{[\bold{x}_f \to \bold{y}_f]}$ tends to 0. Therefore necessarily the characteristic time of every cell $T_C^{(b)}$ must diverges. In fact, if it were not the case for a cache $b$,  none of the contents would be found in the cache $b$ asymptotically, because upward rates tend to zero, while downward rates would not. This would contradict the constraint:
\begin{equation}
 \sum_f\sum_{\bold{x}_f} x^{(b)}_f \pi (\bold{x}_f)=C   \qquad \forall b  \label{caching-loc-constraint}
\end{equation}
imposed by the CTA.
Therefore necessarily $T_C^{(b)}\to \infty$ for every cell $b$.  More precisely we must have 
$T_C^{(b)}=\Theta(\log \frac{1}{q})$  at every cache otherwise we fail to  meet~\eqref{caching-loc-constraint}.
Now we can always select a sequence $\{q_n\}_n$ such that $T_C^{(b)}(q_n)\sim \frac{1}{\gamma_{b}}(\log \frac{1}{q_n})$.

Let us now consider the uniformization of the continuous time MC $\bold{X}_f(t)$ with an arbitrarily high rate $\Lambda_T$ and the corresponding discrete time MC $\bold{X}_f(k)$ with transition probability matrix $P_{f,q}$. \changes{For $q=0$, the set of contents in the cache does not change, each state is an absorbing one and any probability distribution is a stationary probability distribution for $P_{f,0}$. We are rather interested in the asymptotic behaviour of the MC when $q$ converges to $0$. For $q>0$ the MC is finite, irreducible and aperiodic and then admits a unique stationary probability $\bold \pi_{f,q}$.} We call the states $\x_f$ for which $\lim_{q \to 0 } \bold \pi_{f,q}(\x_f) >0$ \emph{stochastically stable}. We are going to characterize such states. 

For what we have said above, it holds that the probability to move from $\x$ to the parent $\y$  is $P_{f,q}(\bold{x}_f,\bold{y}_f) \sim \Lambda_f^{(b_0)} q$, while $P_{f,q}(\bold{y}_f,\bold{x}_f) \sim \Lambda_f^{(b_0)} q^{\Lambda_f^{(b_0)}(\bold{x}_f)/\gamma_{b_0}}$. For each possible transition, we define its \emph{direct resistance} to be the exponent of the parameter $q$, then $r(\bold{x}_f,\bold{y}_f)=1$, $r(\bold{y}_f,\bold{x}_f)=\Lambda_f^{(b_0)}(\bold{x}_f)/\gamma_{b_0}$ and $r(\bold{x}_f,\bold{x}_f)=0$. If a direct transition is not possible between two states, then we consider the corresponding direct resistance to be infinite.  Observe that the higher the resistance, the less likely the corresponding transition. Given a sequence of transitions $(\x^1_f, \x^2_f \dots \x^n_f)$ from state $\x^1_f$ to state $\x^n_f$, we define its resistance to be the sum of the resistances, i.e.~$r(\x^1_f, \x^2_f \dots \x^n_f)=\sum_{i=1}^{n-1} r(\x^i_f, \x^{i+1}_f)$.

The family of Markov chains $\{P_{f,q}\}$ is a \emph{regular perturbation} \cite[properties (6-8)]{young93} and then it is possible to characterize the stochastically stable states as the minimizers of the potential function $V_f(\x_f)$ defined as follows. For each pair of states $\x_f$ and  $\x_f'$ let $R(\x_f, \x_f')$ be the minimum resistance of all the possible sequences of transitions from $\x_f$ to $\x_f'$ (then $R(\x_f, \x_f')\le r(\x_f, \x_f')$). Consider then the full meshed directed weighted graph whose nodes are the possible states of the MC and the weights of the edge  $(\x_f, \x_f')$ is $R(\x_f , \x_f')$. The potential of state $\x_f$ ($V_f(\x_f)$) is defined as the resistance of the minimum weight  in-tree (or anti-arborescence) rooted to $\x_f$. Intuitively the potential is a measure of the general difficulty to reach state $\x_f$ from all the other nodes. From Theorem~4 of~\cite{young93} it follows that $\x_f$ is stochastically stable if and only if its potential is minimal.

For each content $f$ we are then able to characterize which configurations are stochastically stable as $q$ converges to $0$. Moreover, this set of configurations must satisfy the constraint \eqref{caching-loc-constraint} at each base station $b$. We define then the cache configuration $\x = (\x_1, \x_2, \dots \x_F)$ to be \emph{jointly stochastically stable} if 1) for each content $f$ $\x_f$ is stochastically stable, 2) $\x$ satisfies \eqref{caching-loc-constraint} for each $b$. 

The last step in order to prove Proposition~\eqref{p:qlrulazy_convergence_general} is to show that a jointly stochastically stable cache configuration $\x=(\x_1, \x_2, \dots \x_F)$ is locally optimal, i.e.~that changing one content at a given cache does not increase the hit ratio. Without loss of generality, we consider to replace content $f_1$ present at cache $B$ with content $f_2$. Then, the cache allocation $\x$ changes from $\x_{f_1}=(x^{(B)}_{f_1}=1, \x^{(-B)}_{f_1})$ and $\x_{f_2}=(x^{(B)}_{f_2}=0, \x^{(-B)}_{f_2})$ to a new one cache allocation $\x'$, such that $\x'_{f_1}=(x'^{(B)}_{f_1}=0, \x^{(-B)}_{f_1})$ and $\x'_{f_2}=(x'^{(B)}_{f_2}=1, \x^{(-B)}_{f_2})$. 
Let $\eta_f(\x_f)$ denote the hit rate for content $f$ over the whole network under the allocation $\x_f$ and $\eta(\x=(\x_1,\dots \x_F))=\sum_{f=1}^F \eta_f(\x_f)$ the global hit rate across all the contents.
Lemma~\ref{l:hit_rate} below provides a formula for the hit rate $\eta(\x)$, from which we obtain that 
\begin{align}
	\eta(\x)\ge \eta(\x') & \Leftrightarrow \eta_{f_1}(\x_{f_1})+\eta_{f_2}(\x_{f_2}) \nonumber\\
	& \;\;\;\;\;\ge \eta_{f_1}(\x'_{f_1})+\eta_{f_2}(\x'_{f_2})\nonumber\\
			& \Leftrightarrow \Lambda^{(B)}_{f_1}(\x_{f_1}^{(-B)},0) \ge  \Lambda^{(B)}_{f_2}(\x_{f_2}^{(-B)},0)\nonumber\\
			& \Leftrightarrow \Lambda^{(B)}_{f_1}(\x'_{f_1}) \ge  \Lambda^{(B)}_{f_2}(\x_{f_2}).\label{e:rate_ineq}
\end{align}
In order to prove \eqref{e:rate_ineq}, we will show that 
\begin{align}
\Lambda^{(B)}_{f_1}(\x'_{f_1}) & \ge \gamma_B \label{e:rate_ineq1}\\ 
\Lambda^{(B)}_{f_2}(\x_{f_2}) &\le \gamma_B \label{e:rate_ineq2}.
\end{align}

\begin{figure*}%
         \centering
         \subfloat[Proof of Eq.~\eqref{e:rate_ineq1}]{\includegraphics[width=0.45\linewidth]{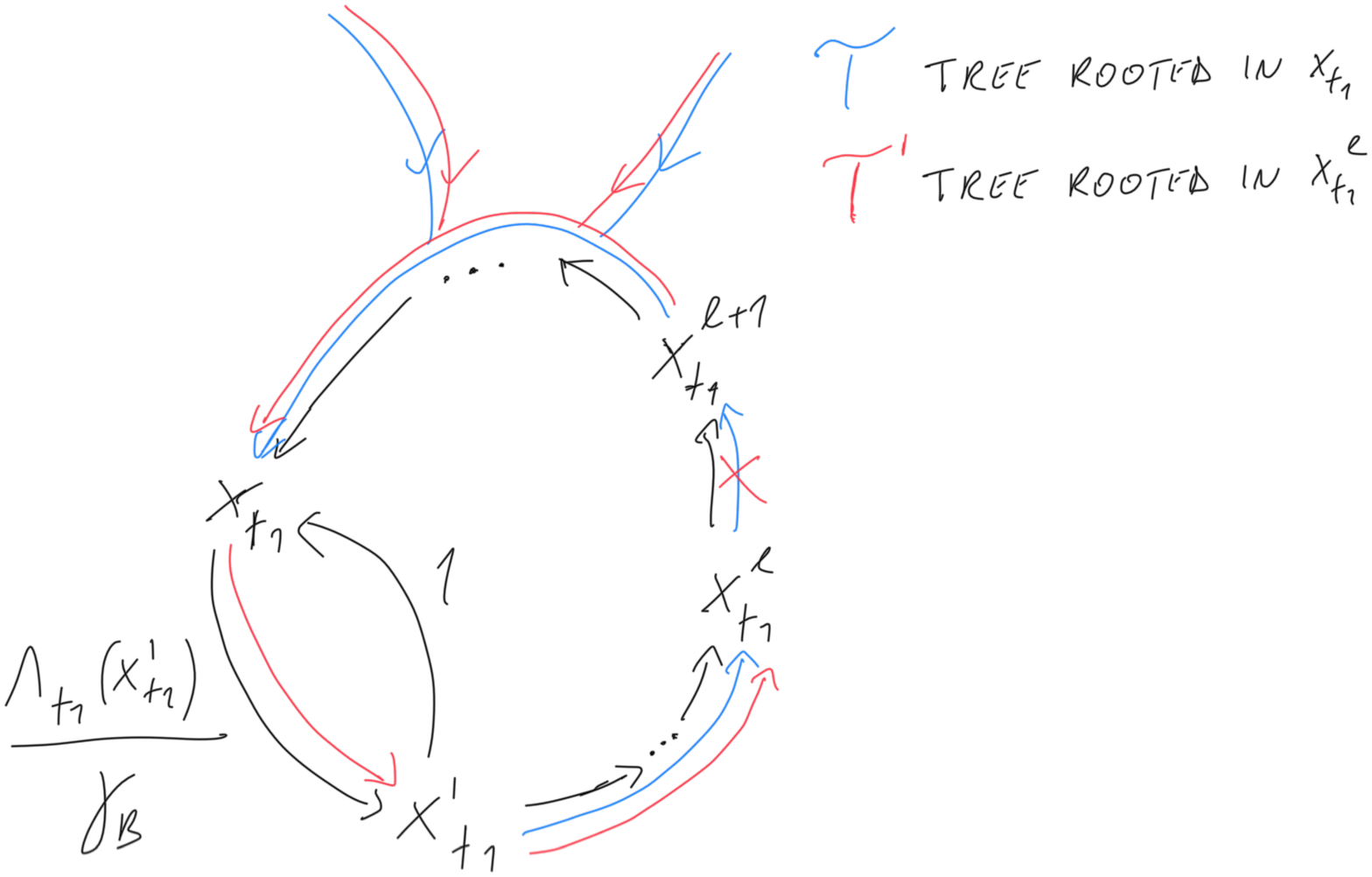}}\hfill
          \subfloat[Proof of Eq.~\eqref{e:rate_ineq2}]{\includegraphics[width=0.45\linewidth]{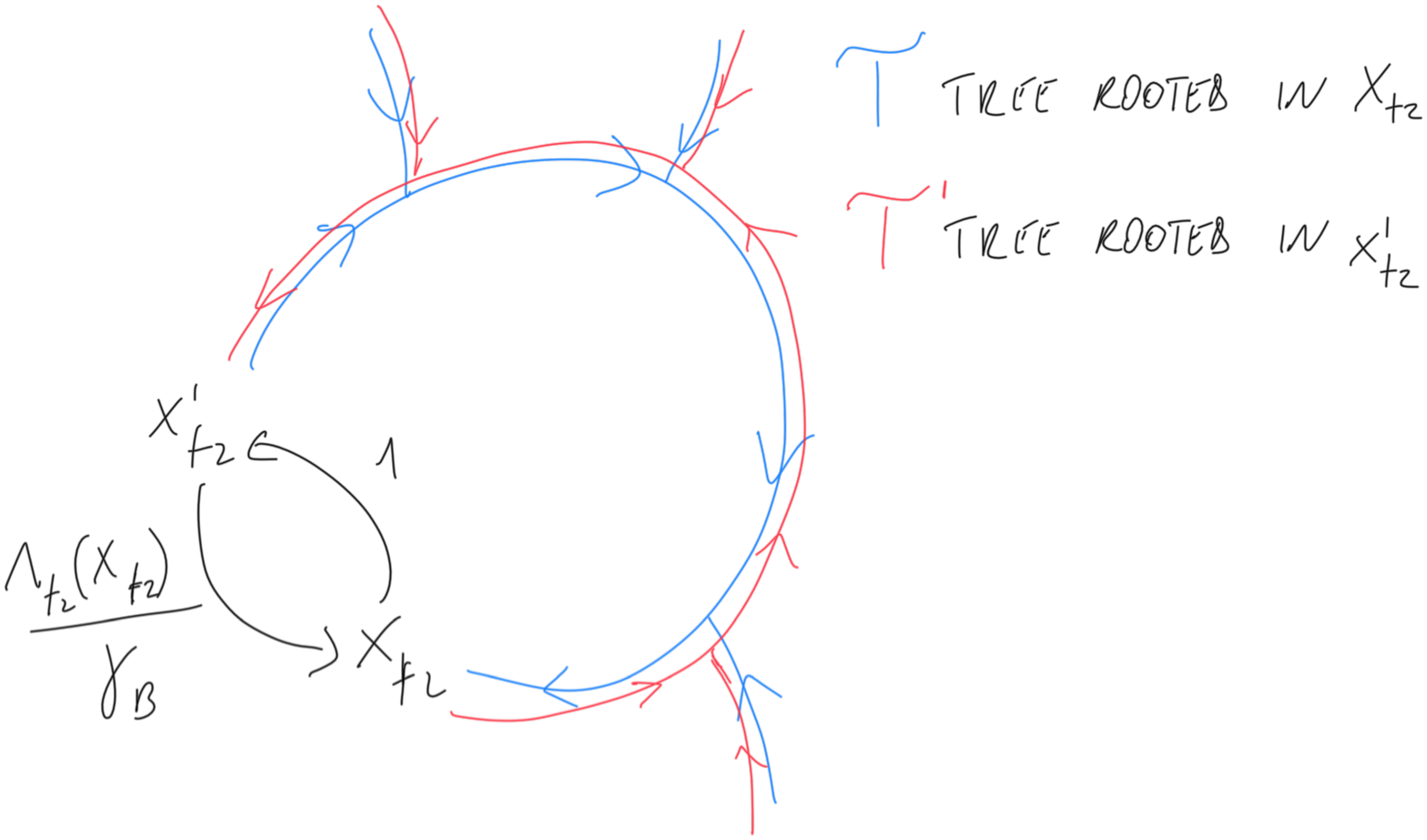}}\hfill
         \caption{Sketch of the constructions used to prove Proposition~\ref{p:qlrulazy_convergence_general}. }
         \label{f:proof_figure1}
       \end{figure*}

The state $\x'_{f_1}$ is a child of $\x_{f_1}$, then $r(\x_{f_1}, \x'_{f_1})=\Lambda^{(B)}_{f_1}(\x'_{f_1})/\gamma_B$. Consider the in-tree $\mathcal T$ rooted in $\x_{f_1}$ with minimal 
resistance 
and let $R(\mathcal T)(=V(\x_{f_1}))$ denote its resistance and 
$(\x^1_{f_1}=\x'_{f_1}, \x^2_{f_1}, \dots \x^k_{f_1}= \x_{f_1})$ be the sequence of transitions in 
$\mathcal T$ from $\x'_{f_1}$ to $\x_{f_1}$. One of these transitions, say it $\x_{f_1}^l \to \x^{l+1}_{f_1} $ corresponds to store the content $f_1$ in the cache $B$ and has resistance $1$. Consider now the in-tree $\mathcal T'$ rooted in $\x_{f_1}^l$ obtained from $\mathcal T$ removing the edge  $(\x_{f_1}^l, \x^{l+1}_{f_1}) $ and adding the edge $(\x_{f_1}, \x'_{f_1})$. Its resistance is $R(\mathcal T')=R(\mathcal T) -1 + \Lambda^{(B)}_{f_1}(\x'_{f_1})/\gamma_B$. From $R(\mathcal T')\ge V(\x_{f_1}^l) \ge V(\x_{f_1}) = R(\mathcal T)$ it follows \eqref{e:rate_ineq1}. A sketch of this construction is in Fig.~\ref{f:proof_figure1}.

The proof of~\eqref{e:rate_ineq2} is slightly more complex.
It is useful to introduce some additional definitions. Given two neighboring states $\x_f$ and $\x'_f$, we say that the transition $\x_f \to \x'_f$ is \emph{dominant} if $r(\x_f,\x'_f) \le r(\x'_f,\x_f)$. Let $\y_f$ be a parent of $\x_f$ with $x_f^{(b)}=0$ and $y_f^{(b)}=1$, we observe that the upward transition $\x_f \to \y_f$ is dominant if and only if $\Lambda^{(b)}_f(\x_f) \ge \gamma_b$. Similarly the downward transition $\y_f \to \x_f$ is dominant if and only if $\Lambda^{(b)}_f(\x_f) \le \gamma_b$. Let us also consider the function of state $\x_f$
\begin{equation}
\label{e:potential2}
\phi(\x_{f})\triangleq\lambda_f \mu\left(\bigcup_{{b \; |\;  x_{f}^{(b)}=1}}S_b\right) - \sum_{b \; |\;  x_{f}^{(b)}=1} \gamma_b.
\end{equation}
Lemma~\ref{l:dominance} guarantees that the function $\phi(.)$ cannot decrease along a dominant transition.

First we prove our result under the assumption that $\gamma_b=\gamma$ for every cell  $b$, then we provide 
the generalization to the most general case.
Let us prove~\eqref{e:rate_ineq2} by contradiction assuming that $\Lambda^{(B)}_{f_2}(\x_{f_2}) > \gamma_B$. In such case $\x'_{f_2} \to \x_{f_2}$ is not a dominant (downward) transition, and  $\phi(\x_{f_2}) < \phi(\x'_{f_2})$. 

Let  now $\mathcal T$ denote the in-tree rooted in $\x_{f_2}$ with minimal resistance and $\Pa=(\x^1_{f_2}= \x'_{f_2}, \x^2_{f_2}, \dots \x^k_{f_2}= \x_{f_2})$ be the sequence of transitions in $\mathcal T$ from $\x'_{f_2}$ to $\x_{f_2}$. At each transition $\x_{f_2}^l \to \x^{l+1}_{f_2}$  only one state variable changes, we denote by $b_l$ the corresponding index, representing the base station at/from which a copy of content $f_2$ is added/removed. 
By construction we have:
\[
 0> \phi(\x_{f_2}) - \phi(\x'_{f_2})= \sum_{1\le l\le k-1}  \phi(\x^{l+1}_{f_2}) - \phi(\x^{l}_{f_2})
\]
Now observe that:
\begin{align}
   \phi(\x^{l+1}_{f_2}) & - \phi(\x^{l}_{f_2})  \nonumber\\
   & =\Lambda_{f_2}^{(b_{l})}(\x^{l}_{f_2})-\gamma_{b_{l}} 
	 = \Lambda_{f_2}^{(b_{l})}(\x^{l+1}_{f_2})-\gamma_{b_{l}}\nonumber\\
   & =\gamma_{b_{l}}[r( \x^{l}_{f_2},\x^{l+1}_{f_2})-r(\x^{l+1}_{f_2}, \x^{l}_{f_2})]
\end{align}
if  transition  $ \x^{l}_{f_2}\to\x^{l+1}_{f_2}  $ is upward, 
and 
\begin{align}
   \phi(\x^{l+1}_{f_2}) & - \phi(\x^{l}_{f_2})=\gamma_{b_{l}}-\Lambda_{f_2}^{(b_{l})}(\x^{l}_{f_2}) \nonumber\\
   & =\gamma_{b_{l}}- \Lambda_{f_2}^{(b_{l})}(\x^{l+1}_{f_2}) \nonumber\\
   & =\gamma_{b_{l}}[r( \x^{l}_{f_2},\x^{l+1}_{f_2})-r(\x^{l+1}_{f_2}, \x^{l}_{f_2})]
\end{align}
if  transition  $ \x^{l}_{f_2}\to\x^{l+1}_{f_2}  $ is downward.
Therefore 
\begin{align}\label{delta-potential}
 \phi(&\x_{f_2}) - \phi(\x'_{f_2})= \nonumber\\
 	&\sum_{1\le l\le k-1}\gamma_{b_{l}}[r( \x^{l}_{f_2},\x^{l+1}_{f_2})-r(\x^{l+1}_{f_2}, \x^{l}_{f_2})]<0
\end{align}
 From which, under the assumption $\gamma_b=\gamma$ for any $b$, we have:
\begin{equation}\label{path-potential}
\sum_{1\le l\le k-1}r( \x^{l}_{f_2},\x^{l+1}_{f_2})>\sum_{1\le l\le k-1}r(\x^{l+1}_{f_2}, \x^{l}_{f_2})
\end{equation}
where the term on the LHS is the total resistance of path $\Pa$  while the term on the RHS is the total resistance of the reverse path 
$\widehat{\Pa}= (\widehat{\x}^1_{f_2}= \x_{f_2}, \widehat{\x}^2_{f_2}= \widehat{\x}^{k-1}_{f_2}, \dots \widehat{\x}^k_{f_2}= \x'_{f_2})$. 

Hence  if we consider the in-tree $\mathcal T'$ routed at $\x'_{f_2}$,  which is obtained from $\mathcal T$ by reverting all the edges of $\Pa$,
i.e.  $\mathcal T'= \mathcal T- \Pa + \widehat{\Pa}$, we obtain that $r(\mathcal T') <r(\mathcal T)$ contradicting the hypothesis. A sketch of this construction is in Fig.~\ref{f:proof_figure2}.

In the most general case, i.e. when $\gamma_b$ are different,  a set of \qlrulazy{} caches all with the same parameter $q$ is not  anymore guaranteed  to be locally optimal (previous proof fails because from \eqref{delta-potential} we cannot deduce \eqref{path-potential}).
However   we can still define a provable locally optimal scheme, if we  
 allow the   adoption of    different parameters $q$ at different cells for the implementation of the local
\qlru\ policy. In particular by   selecting  $q_b= (q_n)^{\gamma_b}$ we can force  characteristic times $T_C^{(b)}$ to be asymptotically equal at different cells.

 \emph{Direct resistances} for our generalized scheme satisfy:
 $r(\bold{x}_f,\bold{y}_f)=\gamma_b$, $r(\bold{y}_f,\bold{x}_f)=\Lambda_f^{(b_0)}(\bold{x}_f)$ when $\y$ is chosen to be a parent of $\x$.
 As a consequence, in this case, we have:
 \[
   \phi(\x_{f}) - \phi(\y_{f}) =r( \x_{f},\y_{f})-r(\x_{f}, \y_{f})
 \]
and 
\[
   \phi(\y_{f}) - \phi(\x_{f}) =r( \y_{f},\x_{f})-r(\y_{f}, \x_{f})
 \]
 Hence by   repeating exactly the same arguments as for the special case $\gamma_b=\gamma$, our generalized scheme  can proved to be  locally optimal.

\end{proof}

\begin{lem}
\label{l:hit_rate}
Given a cache configuration $\x=(\x_1, \dots \x_F)$, under $\qlrulazy$ the hit rate for content $f$ can be calculated as follows
\[\eta_f (\x_f)\triangleq \sum_{b=1}^B \mathbbm{1}(x^{(b)}_f=1) \Lambda^{(b)}_f(x^{(1)}_f, \dots x^{(b-1)}_f, 0, \dots, 0),\]
and the global hit rate is 
\[\eta(\x) \triangleq  \sum_{f=1}^F \eta_f(\x_f).\]
\end{lem}
\begin{proof}
The hit rate for content $f$ is 
\begin{align*}
\eta_f(\x_f) & = \lambda_{f}\mu\left(\bigcup_{b |  x_{f_2}^{(b)}=1}S_b\right) \\
	& = \lambda_f \sum_{b=1}^B \mathbbm{1}(x^{(b)}_f=1) \mu\left(S_b \setminus \bigcup_{b'< b |  x_{f_2}^{(b')}=1}S_b\right)\\
	& = \sum_{b=1}^B \mathbbm{1}(x^{(b)}_f=1) \Lambda^{(b)}_f(x^{(1)}_f, \dots x^{(b-1)}_f, 0, \dots, 0).
\end{align*}
\end{proof}

\begin{lem}
\label{l:dominance}
Given a dominant transition $\x_f \to \y_f$, it holds $\phi(\y) \ge \phi(\x)$.
\end{lem}
\begin{proof}
Let $b'$ be the index at which $\x_f$ and $\y_f$ differ.
If $\x_f \to \y_f$ is an upward dominant transition, then $\Lambda_f^{(b)}(\x_f) \ge \gamma_b$ and it follows:
\begin{align*}
\phi(&\y_{f})  = \lambda_f \mu\left(\bigcup_{{b \; |\;  y_{f}^{(b)}=1}}S_b\right) - \sum_{b \; |\;  y_{f}^{(b)}=1} \gamma_b\\
	& = \lambda_f \mu\left(\bigcup_{{b \; |\;  x_{f}^{(b)}=1}}S_b\right)
		- \sum_{b \; |\;  x_{f}^{(b)}=1} \gamma_b + \Lambda_f^{(b')}(\x_f) - \gamma_{b'}\\
	& \ge \phi(\x_f).
\end{align*}
The proof when $\x_f \to \y_f$ is a downward dominant transition is similar.
\end{proof}

\section{Comparison with [27]}
 \label{a:roberts_experiments}

\begin{figure}%
         \centering
         \includegraphics[width=\myFigureScale\linewidth]{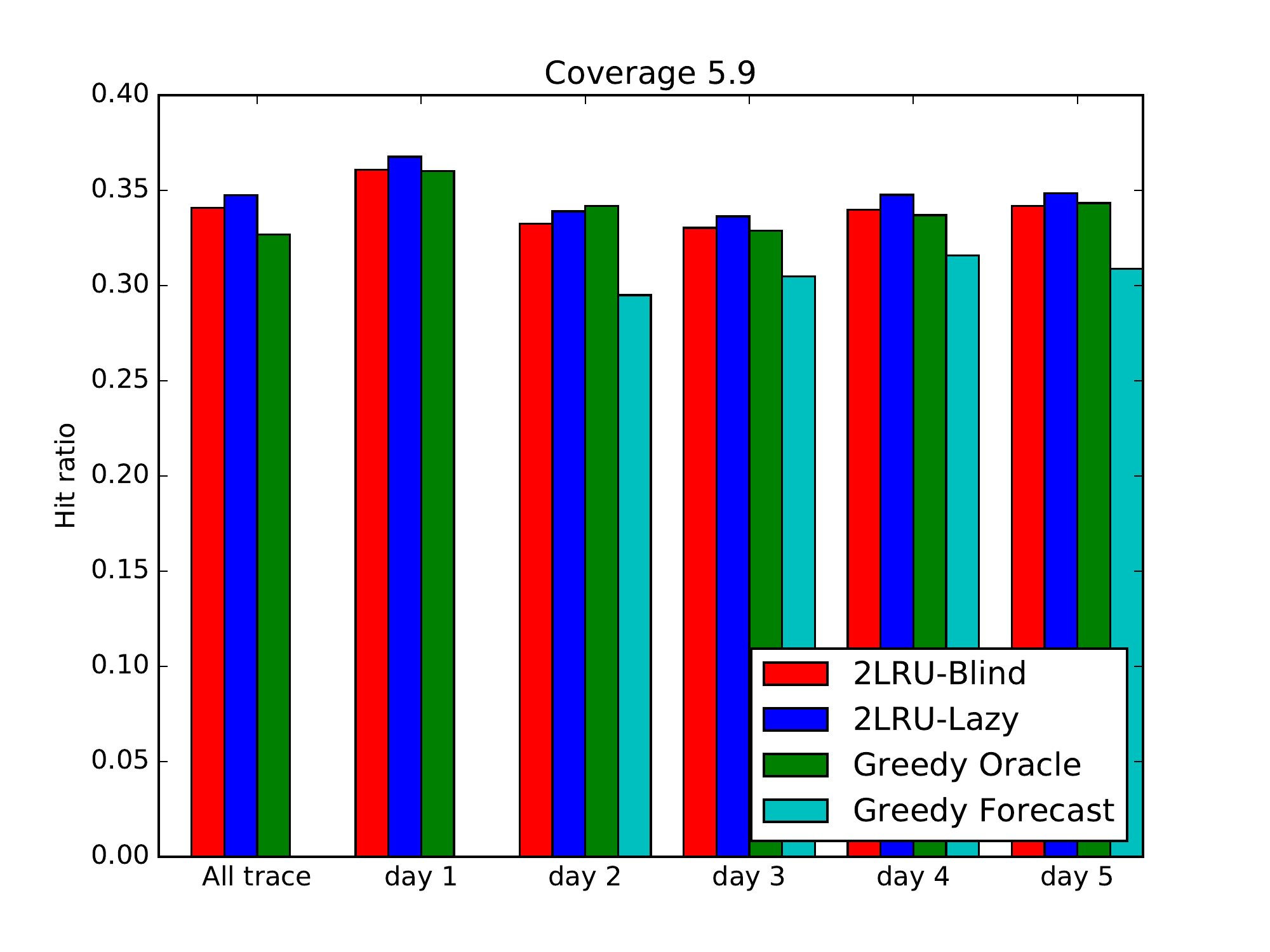}\\
         \caption{\changes{Berlin topology with average coverage $5.9$ and real CDN request trace. Comparison of the different policies over the whole trace and for each of the $5$ days.}}
         \label{f:comparison_day_by_day_sampled}
       \end{figure}

\changes{
As we mentioned in Sec.~\ref{s:realistic}, one of the conclusions of \cite{elayoubi15} is that under the current busy-hour demand (40 Mb/s), an ideal prefetching scheme, based on perfect knowledge of content popularities, performs better than reactive policies. The Akamai request trace we used in Sec.~\ref{s:realistic} corresponds to a busy-hour traffic per cell equal to 800Mbps. We decided then to carry on the same experiments illustrated in Fig.~\ref{f:comparison_day_by_day} sampling our trace by a factor 20. The results are shown in Fig.~\ref{f:comparison_day_by_day_sampled}. The conclusion is that, while \twolrulazy{} is impaired by the lower traffic rate, it still outperforms the static greedy allocation even if clairvoyant day-ahead estimates are available (at least 4 days out of 5). The figure shows also the results for \twolrublind{} (called \lru{} with prefilter in \cite{elayoubi15}): it appears to be overall slightly worse than the clairvoyant static allocation, but it still performs better than the more realistic static allocation based on day-ahead forecast.
} 
\end{document}